\documentclass[11pt]{article}

\usepackage{amsfonts}
\usepackage{amssymb}
\usepackage{amstext}
\usepackage{amsmath}
\usepackage{enumerate}
\usepackage{algorithm}
\usepackage{algorithmic}

\usepackage{graphicx}

\usepackage{bbm}

\usepackage{amsthm}

\usepackage{thmtools}

\usepackage{verbatim} 

\usepackage[usenames,dvipsnames,svgnames,table]{xcolor}
\definecolor{dartmouthgreen}{rgb}{0.05, 0.5, 0.06}
\definecolor{ceruleanblue}{rgb}{0.16, 0.32, 0.75}

\usepackage[colorlinks=true,pdfpagemode=UseNone,citecolor=dartmouthgreen,linkcolor=ceruleanblue,urlcolor=BrickRed,pdfstartview=FitH]{hyperref}

\usepackage{framed}
\usepackage{enumitem}


\newtheorem{theorem}{Theorem}[section]

\newtheorem{lemma}[theorem]{Lemma}
\newtheorem{definition}[theorem]{Definition}

\newtheorem{proposition}[theorem]{Proposition}
\newtheorem{claim}[theorem]{Claim}

\newtheorem*{problem*}{Problem}
\newtheorem{remark}[theorem]{Remark}
\newtheorem*{remark*}{Remark}

\newtheorem{example}[theorem]{Example}

\newcommand{\st}{\mbox{\rm subject to }}

\numberwithin{equation}{section}
\numberwithin{table}{section}

\renewcommand{\preceq}{\preccurlyeq}
\renewcommand{\succeq}{\succcurlyeq}

\renewcommand{\tilde}{\widetilde}

\newcommand{\R}{\ensuremath{\mathbb R}}

\newcommand{\E}[1]{{\mathbb{E}}\left[#1\right]}

\newcommand{\junk}[1]{}

\renewcommand{\l}{\lambda}
\renewcommand{\L}{{\mathcal L}}
\newcommand{\vol}{{\rm vol}}

\newcommand{\norm}[1]{\left\lVert#1\right\rVert}

\newcommand{\vertiii}[1]{{\left\vert\kern-0.25ex\left\vert\kern-0.25ex\left\vert #1 \right\vert\kern-0.25ex\right\vert\kern-0.25ex\right\vert}}

\newcommand{\one}{\ensuremath{\mathbbm{1}}}

\newenvironment{proofof}[1]{{\medbreak\noindent \em Proof of #1.  }}{\hfill\qed\medbreak}

\def\b1{{\bf 1}}
\def\eps{{\epsilon}}

\def\R{\mathbb{R}}

\def\vol{\operatorname{vol}} 
\def\diag{\operatorname{diag}} 

\def\tr{\operatorname{tr}}

\global\long\def\E{\mathbb{E}}

\global\long\def\R{\mathbb{R}}

\newcommand{\inner}[2]{\langle #1, #2 \rangle} 

\DeclareMathOperator{\supp}{supp}

\newcommand{\ip}[2]{\langle #1 , #2\rangle}

\title{Cheeger Inequalities for Directed Graphs and Hypergraphs \\Using Reweighted Eigenvalues}

\author{
Lap Chi Lau\footnote{Cheriton School of Computer Science, University of Waterloo. Supported by NSERC Discovery Grant. 
},~~~~~
Kam Chuen Tung\footnote{Cheriton School of Computer Science, University of Waterloo. Supported by NSERC Discovery Grant. 
},~~~~~
Robert Wang\footnote{Cheriton School of Computer Science, University of Waterloo. Supported by NSERC Discovery Grant and Canada Graduate Scholarship. 
}}

\date{}

\begin{document}

\begin{titlepage}
\def\thepage{}
\thispagestyle{empty}

\maketitle

\begin{abstract}
We derive Cheeger inequalities for directed graphs and hypergraphs using the reweighted eigenvalue approach that was recently developed for 
vertex expansion in undirected graphs~\cite{OZ22,KLT22,JPV22}. 
The goal is to develop a new spectral theory for directed graphs 
and an alternative spectral theory for hypergraphs.

The first main result is a Cheeger inequality relating the vertex expansion $\vec{\psi}(G)$ of a directed graph $G$ to the vertex-capacitated maximum reweighted second eigenvalue $\vec{\lambda}_2^{v*}$ that
\[
\vec{\lambda}_2^{v*} \lesssim \vec{\psi}(G) \lesssim \sqrt{\vec{\lambda}_2^{v*} \cdot \log \frac{\Delta}{\vec{\lambda}_2^{v*}}}
\]
where $\Delta$ is the maximum degree of $G$.
This provides a combinatorial characterization of the fastest mixing time of a directed graph by vertex expansion, and builds a new connection between reweighted eigenvalued, vertex expansion, and fastest mixing time for directed graphs.
 
The second main result is a stronger Cheeger inequality relating the edge conductance $\vec{\phi}(G)$ of a directed graph $G$ to the edge-capacitated maximum reweighted second eigenvalue $\vec{\lambda}_2^{e*}$ that
\[
\vec{\lambda}_2^{e*} \lesssim \vec{\phi}(G) \lesssim \sqrt{\vec{\lambda}_2^{e*} \cdot \log \frac{1}{\vec{\lambda}_2^{e*}}}.
\]
This provides a certificate for a directed graph to be an expander and a spectral algorithm to find a sparse cut in a directed graph,
playing a similar role as Cheeger's inequality in certifying graph expansion and in the spectral partitioning algorithm for undirected graphs.

We also use this reweighted eigenvalue approach to derive the improved Cheeger inequality for directed graphs, and furthermore to derive several Cheeger inequalities for hypergraphs that match and improve the existing results in~\cite{Lou15,CLTZ18}. 
These are supporting results that this provides a unifying approach to lift the spectral theory for undirected graphs to more general settings. 
\end{abstract}

\end{titlepage}

\thispagestyle{empty}


\newpage

\section{Introduction}

Cheeger's inequality~\cite{Che70,AM85,Alo86,Chu97} is a fundamental result in spectral graph theory that connects the edge expansion property of an undirected graph $G=(V,E)$ to the second eigenvalue of its associated matrix:
\vspace*{-3mm}
\begin{equation} \label{e:Cheeger}
\frac{\lambda_2}{2} \leq \phi(G) \leq \sqrt{2\lambda_2}
\end{equation}
where $\phi(G)$ is the edge conductance of $G$ and $\lambda_2$ is the second smallest eigenvalue of its normalized Laplacian matrix\footnote{See \autoref{sec:prelim} for various definitions that are not stated in this introduction.}.
There are two important applications of Cheeger's inequality.
One is to use the second eigenvalue to study expander graphs~\cite{HLW06} and its eigenvector for graph partitioning~\cite{SM00,Lux07}.
The other is to use the edge conductance to bound the mixing time of random walks~\cite{AF02,LP17}.
Together, Cheeger's inequality connects the second eigenvalue, edge conductance, and mixing time.
More recently, the spectral theory for undirected graphs is enriched by several interesting generalizations of Cheeger's inequality~\cite{Tre09,ABS10,LOT12,LRTV12,KLLOT13}, which establish further connections between edge expansion properties of the graph to other eigenvalues of its normalized Laplacian matrix.

In contrast, 
the spectral theory for directed graphs has not been nearly as well developed. 
One issue is that the Laplacian matrix of a directed graph is not Hermitian, and so its eigenvalues are not necessarily real numbers.
There are formulations~\cite{Fil91,Chu05,GM17,LL15} 
that associate certain Hermitian matrices to a directed graph,
and use the second eigenvalue of these matrices to bound the mixing time of random walks~\cite{Fil91,Chu05} (see \autoref{sec:related-work} for details).
But, to our knowledge,
there are no known formulations that relate the expansion properties of a directed graph to the eigenvalues of an associated matrix\footnote{The only formulation that we know about expansion properties of a directed graph is a {\em nonlinear} Laplacian operator in~\cite{Yos16,Yos19}.  See \autoref{sec:related-work} for details.}.
The main goal of this paper is to provide such formulations using ``reweighted eigenvalues'' and to develop a spectral theory for directed graphs that is comparable to that for undirected graphs.

The notion of reweighted eigenvalue for undirected graphs was first formulated in~\cite{BDX04} for studying the fastest mixing time problem on reversible Markov chains.
In this formulation, we are given an undirected graph $G=(V,E)$, 
and the task is to find a reweighted graph $G'=(V,E,w)$ with edge weight $w(uv)$ for $uv \in E$ and weighted degree one for each vertex, that maximizes the second eigenvalue $\lambda_2^*$ of its normalized Laplacian matrix.
It was known~\cite{Roc05} that the vertex expansion $\psi(G)$ is an upper bound on $\lambda_2^*$, but only recently~\cite{OZ22,KLT22,JPV22} was it established that there is a Cheeger-type inequality
relating these two quantities:
\begin{equation} \label{e:vertex-Cheeger}
\lambda_2^* \lesssim \psi(G) \lesssim \sqrt{\lambda_2^* \cdot \log \Delta}
\end{equation}
where $\Delta$ is the maximum degree of a vertex in $G$.
This inequality connects the reweighted second eigenvalue and vertex expansion and fastest mixing time, in a similar way that Cheeger's inequality connects the second eigenvalue and edge conductance and mixing time.
This reweighted eigenvalue approach was extended in~\cite{KLT22} to develop a spectral theory for undirected vertex expansion,
by proving that several generalizations of Cheeger's inequality~\cite{Tre09,LOT12,LRTV12,KLLOT13} have close analogs in connecting vertex expansion properties to other reweighted eigenvalues. 

\subsection{Our Results}
We formulate reweighted eigenvalues for directed graphs and hypergraphs.
The main idea is to reduce the study of expansion properties in directed graphs and hypergraphs to the basic setting of edge conductances in undirected graphs.
We show that this provides an intuitive and unifying approach to lift the spectral theory for undirected graphs to more general settings.

\subsubsection{Cheeger Inequality for Directed Vertex Expansion}

Classical spectral theory connects (i) undirected edge conductance, (ii) second eigenvalue, and (iii) mixing time of random walks on undirected graphs.
We present a new spectral formulation that connects (i) directed vertex expansion, (ii) reweighted second eigenvalue, and (iii) fastest mixing time of random walks on directed graphs.

\begin{definition}[Directed Vertex Expansion] \label{def:directed-vertex-expansion}
Let $G=(V,E)$ be a directed graph and $\pi:V \to \R_{\geq 0}$ be a weight function on the vertices.
For a subset $S \subseteq V$, let $\partial^+(S) := \{v \notin S \mid \exists u \in S \textrm{~with~} uv \in E \}$ be the set of out-neighbors of $S$, and $\pi(S) := \sum_{v \in S} \pi(v)$ be the weight of $S$.
The directed vertex expansion of a set $S \subseteq V$ and of the graph $G$ are defined as\footnote{When specialized to undirected graphs (by considering the bidirected graph), the current definitions are slightly different from that in~\cite{OZ22,KLT22}; see \autoref{sec:prelim}. 
We remark that the two definitions of $\psi(G)$ are within a factor of $2$ of each other.
The current definitions have the advantages that $\psi(S) \leq 1$ and are more convenient in the proofs.}
\[
\vec{\psi}(S) := \frac{\min\big\{ \pi\big(\partial^+(S)\big), \pi\big(\partial^+(\overline{S})\big) \big\}}{\min\big\{\pi(S),\pi(\overline{S})\big\}}
\quad {\textrm and} \quad
\vec{\psi}(G) := \min_{\emptyset \neq S \subset V} \vec{\psi}(S).
\]
where $\overline{S} := V-S$ is the complement set of $S$.
Note that $\vec{\psi}(S) \leq 1$
for all $S \subseteq V$ as $\partial^+(\overline{S}) \subseteq S$.
\end{definition}

To certify that a directed graph $G=(V,E)$ has large vertex expansion,
our idea is to find the best reweighted {\em Eulerian} subgraph $G'=(V,E,w)$ of $G$ with edge weight $w(uv)$ for $uv \in E$ and weighted degrees $\sum_{u \in V} w(uv) = \sum_{u \in V} w(vu) = \pi(v)$ for $v \in V$, and then use the edge conductance of $G'$ as a lower bound on the vertex expansion of $G$.
Since the weighted directed graph $G'$ is Eulerian,
the edge conductance of $G'$ is equal to the edge conductance of the underlying undirected graph $G''$ with edge weight $w''(uv) = \frac{1}{2} \big(w(uv) + w(vu)\big)$. 
Now, as the graph $G''$ is undirected, 
we can use Cheeger's inequality to lower bound the edge conductance of $G''$ by the second smallest eigenvalue of its normalized Laplacian matrix.
This leads to the following formulation of the reweighted second eigenvalue for directed vertex expansion
(see \autoref{prop:directed-vertex-easy} for more about this reduction).

\begin{definition}[Maximum Reweighted Spectral Gap with Vertex Capacity Constraints] \label{def:directed-vertex-primal}
Given a directed graph $G=(V,E)$ and a weight function $\pi: V \to \R_{\geq 0}$,
the maximum reweighted spectral gap with vertex capacity constraints is defined as
\begin{align*}
{\vec{\lambda}_2^{v*}}(G) ~:=~ \max_{A \geq 0} &~~~ \lambda_2\bigg( I - \Pi^{-\frac12} \Big(\frac{A+A^T}{2}\Big) \Pi^{-\frac12}\bigg) & 
\\
\st &~~~ A(u,v) = 0 & & \forall uv \notin E
\\
&~~~ \sum_{v \in V} A(u,v) = \sum_{v \in V} A(v,u) & & \forall u \in V
\\
&~~~ \sum_{v \in V} A(u,v) = \pi(u) & & \forall u \in V
\end{align*}
where $A$ is the adjacency matrix of the reweighted Eulerian subgraph and
$\Pi := \diag(\pi)$ is the diagonal degree matrix of $A$.
Then $\frac12 (A+A^T)$ is the adjacency matrix of the underlying undirected graph of the reweighted Eulerian subgraph, $\L := I - \frac12 \Pi^{-1/2} (A+A^T) \Pi^{-1/2}$ is its normalized Laplacian matrix, and $\lambda_2(\L)$ is the second smallest eigenvalue of $\L$.

To ensure that the optimization problem for $\lambda_2^{v}(G)$ is always feasible, we assume that the graph has a self-loop at each vertex.
In the context of Markov chains, this corresponds to allowing a non-zero holding probability on each vertex.
\end{definition}

Our first main result is a Cheeger-type inequality that relates $\vec{\lambda}_2^{v*}(G)$ and $\vec{\psi}(G)$, proving that the directed vertex expansion is large if and only if the reweighted eigenvalue is large. 

\begin{theorem}[Cheeger Inequality for Directed Vertex Expansion] \label{thm:directed-vertex-expansion}
For any directed graph $G=(V,E)$ and any weight function $\pi: V \to \R_{\geq 0}$,
\[
\vec{\lambda}_2^{v*}(G) \lesssim \vec{\psi}(G) \lesssim \sqrt{\vec{\lambda}_2^{v*}(G) \cdot  \log \frac{\Delta}{\vec{\psi}(G)} }
\lesssim \sqrt{\vec{\lambda}_2^{v*}(G) \cdot  \log \frac{\Delta}{\vec{\lambda}_2^{v*}(G)} },
\]
where $\Delta$ is the maximum (unweighted) degree of a vertex of $G$.
\end{theorem}

Since directed vertex expansion is more general than undirected vertex expansion and \eqref{e:vertex-Cheeger} is tight up to a constant factor~\cite{KLT22}, we know that the $\log \Delta$ term in \autoref{thm:directed-vertex-expansion} is necessary.  
But we do not know whether the $\log(1/\vec{\psi}(G))$ term in \autoref{thm:directed-vertex-expansion} is necessary or not. \\

{\bf The Fastest Mixing Time Problem}:
The notion of reweighted eigenvalue for undirected graphs was first formulated in~\cite{BDX04} for studying the fastest mixing time problem on {\em reversible} Markov chains.
It turns out that the reweighted eigenvalue $\vec{\lambda}_2^{v*}(G)$ in \autoref{def:directed-vertex-primal} can be used to study the fastest mixing time problem on {\em general} Markov chains.

\begin{definition}[Fastest Mixing Time on General Markov Chain] \label{def:fastest-mixing}
Given a directed graph $G=(V,E)$ and a probability distribution $\pi$ on $V$,
the fastest mixing time problem is defined as
\begin{align*}
\tau^*(G) ~:=~ \min_{P \geq 0} &~~~ \tau(P) & 
\\
\st &~~~ P(u,v) = 0 & & \forall uv \notin E
\\
&~~~ \sum_{v \in V} P(u,v) = 1 & & \forall u \in V
\\
&~~~ \sum_{u \in V} \pi(u) \cdot P(u,v) = \pi(v) & & \forall v \in V
\end{align*}
where $P$ is the transition matrix of the Markov chain.
The constraints are to ensure that $P$ has nonzero entries only on the edges of $G$, that $P$ is a row stochastic matrix, and that the stationary distribution of $P$ is $\pi$.
The objective is to minimize the mixing time $\tau(P)$ to the stationary distribution $\pi$; 
see \autoref{sec:prelim} for definitions of random walks and mixing times.
\end{definition}

For the fastest mixing time problem on reversible Markov chains, we are given an undirected graph $G=(V,E)$ and a probability distribution $\pi$, and the last set of constraints in \autoref{def:fastest-mixing} is replaced by the stronger requirement that $\pi(u) \cdot P(u,v) = \pi(v) \cdot P(v,u)$ for all $uv \in E$.
With this stronger requirement, $P$ has real eigenvalues and it is well known that 
$\tau(P) \lesssim \frac{1}{1-\alpha_2(P)} \cdot \log( \frac{1}{\pi_{\min}}),$ 
where $\alpha_2(P)$ is the second largest eigenvalue of $P$ and $\pi_{\min} := \min_{v \in V} \pi(v)$.
Thus, the reweighted eigenvalue formulation in~\cite{BDX04} is to find such a transition matrix $P$ that maximizes the spectral gap $1-\alpha_2(P)$, which can be solved by a semidefinite program and can be used as a proxy to upper bounding the fastest mixing time.

For general Markov chains, $P$ may have complex eigenvalues, 
and there was no known efficient formulation for the fastest mixing time problem.
We observe that the reweighted spectral gap $\vec{\lambda}_2^{v*}(G)$ in \autoref{def:directed-vertex-primal} provides such a formulation through the results in~\cite{Fil91,Chu05}.
An interesting consequence of \autoref{thm:directed-vertex-expansion} is a combinatorial characterization of the fastest mixing time of general Markov chains, showing that small directed vertex expansion is the only obstruction of fastest mixing time.

\begin{theorem}[Fastest Mixing Time and Directed Vertex Expansion] \label{thm:fastest-mixing}
For any directed graph $G = (V,E)$ with maximum total degree $\Delta$, and for any probability distribution $\pi$ on $V$,
\[
\frac{1}{\vec{\psi}(G)} \cdot \frac{1}{\log(1/\pi_{\min})}\lesssim \tau^*(G) \lesssim \frac{1}{\vec{\psi}(G)^2} \cdot \log \frac{\Delta}{\vec{\psi}(G)} \cdot \log \frac{1}{\pi_{\min}}.
\]
\end{theorem}

Together, \autoref{thm:directed-vertex-expansion} and \autoref{thm:fastest-mixing} connect the reweighted second eigenvalue, directed vertex expansion, and fastest mixing time on directed graphs, 
in a similar way that classical spectral graph theory connects the second eigenvalue, undirected edge conductance, and mixing time on undirected graphs.

\subsubsection{Cheeger Inequality for Directed Edge Conductance}

Two key applications of Cheeger's inequality are to use the second eigenvalue to certify whether an undirected graph is an expander graph, and to provide a spectral algorithm for graph partitioning that is useful in many areas.
We present a new inequality for directed graphs for these purposes.

\begin{definition}[Directed Edge Conductance~\cite{Yos16,Yos19}] \label{def:directed-edge-conductance}
Let $G=(V,E)$ be a directed graph and $w : E \to \R_{\geq 0}$ be a weight function on the edges.
For a subset $S \subseteq V$, let $\delta^+(S) := \{uv \in E \mid u \in S \textrm{~and~} v \notin S\}$ be the set of outgoing edges of $S$ and $w(\delta^+(S))$ be the total edge weight on $\delta^+(S)$,
and $\vol_w(S) := \sum_{v \in S} \sum_{u \in V} (w(uv) + w(vu))$ be the volume of $S$.
The directed edge conductance of a set $S \subseteq V$ and of the graph $G$ are defined as 
\[
\vec{\phi}(S) := \frac{\min\big\{ w\big(\delta^+(S)\big), w\big(\delta^+(\overline{S})\big) \big\}}{\min\big\{\vol_w(S),\vol_w(\overline{S})\big\}}
\quad {\textrm and} \quad
\vec{\phi}(G) := \min_{\emptyset \neq S \subset V} \vec{\phi}(S).
\]
\end{definition}

We use the same approach to prove a Cheeger-type inequality for directed edge conductance\footnote{The reader may wonder whether it is possible to reduce directed edge conductance to directed vertex expansion, and use \autoref{thm:directed-vertex-expansion} to obtain a Cheeger-type inequality for directed edge conductance.
This is indeed possible, but the result obtained in this way will have a dependency on the maximum total degree $\Delta$ as in \autoref{thm:directed-vertex-expansion}, while the result that we present in \autoref{thm:directed-edge-conductance} has no such dependency.}.
To certify that a directed graph $G=(V,E)$ with a weight function $w : E \to \R_{\geq 0}$ has large edge conductance,
we find the best reweighted Eulerian subgraph $G'$ with edge weight $w'(u,v) \leq w(u,v)$ for each $uv \in E$, and use the edge conductance of $G'$ (with respect to the volumes using $w$) to provide a lower bound on the edge conductance of $G$.
Then, the edge conductance of $G'$ is reduced to the edge conductance of the underlying undirected graph $G''$ with edge weight $w''(uv) = \frac12(w'(uv)+w'(vu))$, and the second smallest eigenvalue of the normalized Laplacian matrix of $G''$ is used to provide a lower bound on the edge conductance of $G''$.
See \autoref{prop:directed-edge-easy} for a proof.

\begin{definition}[Maximum Reweighted Spectral Gap with Edge Capacity Constraints] \label{def:directed-edge-primal}
Given a directed graph $G=(V,E)$ and a weight function $w: E \to \R_{\geq 0}$,
the maximum reweighted spectral gap with edge capacity constraints is defined as
\begin{align*}
{\vec{\lambda}_2^{e*}}(G) ~:=~ \max_{A \geq 0} &~~~ \lambda_2\bigg( D^{-\frac12} \Big(D_A - \frac{A+A^T}{2}\Big) D^{-\frac12}\bigg) & 
\\
\st &~~~ A(u,v) = 0 & & \forall uv \notin E
\\
&~~~ \sum_{v \in V} A(u,v) = \sum_{v \in V} A(v,u) & & \forall u \in V
\\
&~~~ A(u,v) \leq w(uv) & & \forall uv \in E
\end{align*}
where $A$ is the adjacency matrix of the reweighted Eulerian subgraph, $D_A$ is the diagonal degree matrix of $(A+A^T)/2$ with $D_A(v,v) = \sum_{u \in V} \frac12 (A(u,v)+A(v,u))$, and
$D$ is the diagonal degree matrix of $G$ with $D(v,v) = \sum_{u \in V} (w(uv) + w(vu))$ equal to the total weighted degree of $v$ in $G$.
\end{definition}

Our second main result is a stronger Cheeger-type inequality that relates $\vec{\lambda}_2^{e*}(G)$ and $\vec{\phi}(G)$.

\begin{theorem}[Cheeger Inequality for Directed Edge Conductance] \label{thm:directed-edge-conductance}
For any directed graph $G=(V,E)$ and any weight function $w: E \to \R_{\geq 0}$,
\[
\vec{\lambda}_2^{e*}(G) \lesssim \vec{\phi}(G) \lesssim \sqrt{\vec{\lambda}_2^{e*}(G) \cdot  \log \frac{1}{\vec{\phi}(G)} }
\lesssim \sqrt{\vec{\lambda}_2^{e*}(G) \cdot  \log \frac{1}{\vec{\lambda}_2^{e*}(G)} }.
\]
\end{theorem}

An important point about \autoref{thm:directed-edge-conductance} is that there is no dependence on the maximum degree of $G$ as in \autoref{thm:directed-vertex-expansion} or on the number of vertices of $G$ as in~\cite{ACMM05,Yos19}. 
As a consequence, $\vec{\lambda}_2^{e*}(G)$ is a polynomial time-computable quantity that can be used to certify whether a directed graph has constant edge conductance.
This is similar to the role of the second eigenvalue in Cheeger's inequality to certify whether an undirected graph has constant edge conductance.

Also, as in the proof of Cheeger's inequality, the proof of \autoref{thm:directed-edge-conductance} provides a polynomial time algorithm to return a set $S$ with $\vec{\phi}(S) \leq \sqrt{\vec{\phi}(G) \log 1/\vec{\phi}(G)}$.
Since many real-world networks are directed (see~\cite{Yos16}),
we hope that this ``spectral'' algorithm will find applications in clustering and partitioning for directed graphs, as the classical spectral partitioning algorithm does in clustering and partitioning for undirected graphs~\cite{SM00,Lux07}.

\subsubsection{Generalizations of Cheeger Inequality for Directed Graphs}

For undirected graphs, there are several interesting generalizations of Cheeger's inequality: Trevisan's result~\cite{Tre09} that relate $\lambda_n$ to bipartite edge conductance, the higher-order Cheeger's inequality~\cite{LOT12,LRTV12} that relates $\lambda_k$ to $k$-way edge conductance, and the improved Cheeger's inequality~\cite{KLLOT13} that relates $\lambda_2$ and $\lambda_k$ to edge conductance.
Using reweighted eigenvalues for vertex expansion, close analogs of these results were obtained in~\cite{KLT22}, relating $\lambda_n^*$ to bipartite vertex expansion, $\lambda_k^*$ to $k$-way vertex expansion, and $\lambda_2^*$ and $\lambda_k^*$ to vertex expansion. 

We study whether there are close analogs of these results for directed graphs, using reweighted eigenvalues for directed vertex expansion in \autoref{def:directed-vertex-expansion} and directed edge conductance in \autoref{def:directed-edge-conductance}.
Perhaps surprisingly, we show that the natural analogs of Trevisan's result and higher-order Cheeger's inequality do not hold, but we obtain analogs of the improved Cheeger's inequality for directed vertex expansion and directed edge conductance.
See \autoref{sec:supporting-results} for these results.

\subsubsection{Cheeger Inequalities for Hypergraph Edge Conductance}
\label{sec:intro-cheeger-hypergraph}
We also formulate reweighted eigenvalues for hypergraphs and use them to derive Cheeger-type inequalities for hypergraphs,
as supporting results that reweighted eigenvalues provide a unifying approach to study expansion properties in different settings.

\begin{definition}[Hypergraph Edge Conductance~\cite{Lou15,CLTZ18}] \label{def:hypergraph-edge-conductance}
Let $H=(V,E)$ be a hypergraph and $w : E \to \R_{\geq 0}$ be a weight function on the hyperedges.
For a subset $S \subseteq V$, let $\delta(S) := \{e \in E \mid e \cap S \neq \emptyset {\rm~and~} e \cap \overline{S} \neq \emptyset\}$ be the edge boundary of $S$ and $w(\delta(S))$ be the total edge weight of $\delta(S)$,
and let $\vol_w(S) := \sum_{v \in S} \sum_{e: v \in e} w(e)$ be the volume of $S$.
The hypergraph edge conductance of a set $S \subseteq V$ and of the graph $G$ are defined as 
\[
\phi(S) := \frac{ w\big(\delta(S)\big) }{\min\big\{\vol_w(S),\vol_w(\overline{S})\big\}}
\quad {\textrm and} \quad
{\phi}(H) := \min_{\emptyset \neq S \subset V} {\phi}(S).
\]
\end{definition}

The idea is simply to consider the ``clique-graph'' of the hypergraph $H$, and find the best reweighted subgraph of the clique-graph to certify the edge conductance of $H$, subject to the constraint that the total weight of the ``clique-edges'' of a hyperedge $e$ is bounded by $w(e)$.

\begin{definition}[Maximum Reweighted Spectral Gap for Hypergraphs] \label{def:hypergraph-edge-primal}
Given a hypergraph $H=(V,E)$ and a weight function $w: E \to \R_{\geq 0}$,
the maximum reweighted spectral gap for $H$ is defined as
\begin{align*}
\gamma_2^*(H) ~:=~ \max_{A \geq 0} &~~~ \lambda_2\Big( D^{-\frac12} \big(D_A-A\big) D^{-\frac12}\Big) & 
\\
\st
&~~~ \sum_{u,v \in e} c(u,v,e) \leq w(e) & & \forall e \in E
\\
&~~~ A(u,v) = \sum_{e \in E: u,v \in e} c(u,v,e) & & \forall u,v \in V.
\end{align*}
In this formulation, there is a clique-edge variable $c(u,v,e)$ for each pair of vertices $u,v$ in a hyperedge $e$, with the constraints that the total weight of the clique-edges in $e$ is bounded by $w(e)$.
Then, $A$ is the adjacency matrix of the reweighted subgraph of the clique-graph with edge weight $A(u,v)$ equal to the sum of the weight of the clique-edges involving $u$ and $v$, $D_A$ is the diagonal degree matrix of $A$ with $D_A(v,v) = \sum_{u \in V} A(u, v)$,
and $D$ is the diagonal degree matrix of $H$ with $D(v,v) = \sum_{e \in E : v \in e} w(e)$ equal to the weighted degree of $v$ in $H$.
\end{definition}

There is a spectral theory for hypergraphs based on a continuous time diffusion process with several Cheeger-type inequalities proven~\cite{Lou15,CLTZ18}.
We show that the reweighted eigenvalue approach can be used to provide a simpler and more intuitive way to obtain similar results.

\begin{theorem}[Cheeger Inequality for Hypergraph Edge Conductance] \label{thm:hypergraph-edge-conductance}
For any hypergraph $H=(V,E)$ and any weight function $w: E \to \R_{\geq 0}$,
\[
\gamma_2^*(H) \lesssim \phi(H) \lesssim \sqrt{\gamma_2^{*}(H) \cdot  \log(r)  }
\]
where $r$ is the maximum size of a hyperedge of $H$.
\end{theorem}

We also obtain generalizations of Cheeger's inequalities for hypergraphs using other reweighted eigenvalues such as $\gamma_k^*(H)$ and a new result about improved Cheeger inequality for hypergraphs.
We will mention these results and compare the two approaches in \autoref{sec:related-work}.

\vspace{-1mm}

\subsection{Techniques}

Conceptually, our contribution is to come up with new spectral formulations for expansion properties in directed graphs and hypergraphs, and to show that the reweighted eigenvalue approach provides a unifying method to reduce expansion problems in more general settings to the basic setting of edge conductances in undirected graphs.

Technically, the proofs are based on the framework developed in~\cite{OZ22,KLT22,JPV22} in relating reweighted eigenvalues to undirected vertex expansion in \eqref{e:vertex-Cheeger}.
We briefly describe this framework and then highlight some new elements in the proofs for directed graphs.
There are two main steps in proving \eqref{e:vertex-Cheeger}.
The first step is to construct the dual SDP for the reweighted eigenvalue, and to do random projection to obtain a $1$-dimensional solution to the dual program.
The second step is to analyze the threshold rounding algorithm for the $1$-dimensional solution.

For directed graphs, we identify a key parameter for our analysis.

\begin{definition}[Asymmetric Ratio of Directed Graphs] \label{def:asymmetric-ratio}
Given an edge-weighted graph $G = (V, E, w)$, the asymmetric ratio of a set $S \subseteq V$ and of the graph $G$ are defined as
\[
\alpha(S) := \frac{w(\delta^+(S))}{w(\delta^+(\overline{S}))} 
\quad {\rm and} \quad
\alpha(G) := \max_{\emptyset \neq S \subset V} \alpha(S).
\]
Given a vertex-weighted graph $G = (V, E, \pi)$,
we define the $\pi$-induced weight of an edge $uv \in E$ as $w_{\pi}(uv) = \min\{\pi(u),\pi(v)\}$,
and the asymmetric ratio of a set $S \subseteq V$ and of the graph are defined as above using the edge weight function $w_{\pi}$. 
\end{definition}

We note that the asymmetric ratio of an edge-weighted graph was defined in~\cite{EMPS16} with the name ``$\alpha$-balanced'' and was used in the analysis of oblivious routing in directed graphs.
The asymmetric ratio is a measure of how close a directed graph is to an undirected graph for our purpose, as when $\alpha(G)=1$ the directed graph is Eulerian and so its edge conductance is the same as the edge conductance of the underlying undirected graph.

This parameter is defined to satisfy two useful properties.
The first is that it can be used to prove more refined Cheeger inequalities that
\begin{equation} \label{e:refined}
\vec{\phi}(G) \leq \sqrt{\vec{\lambda}_2^{e*}(G) \cdot \log \alpha(G)} 
\quad {\rm and} \quad 
\vec{\psi}(G) \leq \sqrt{\vec{\lambda}_2^{v*}(G) \cdot \log \big( \Delta \cdot \alpha(G) \big)}.
\end{equation}
The second is that it can be related to the directed edge conductance and directed vertex expansion such that $\alpha(G) \leq 1/\vec{\phi}(G)$ in \autoref{lem:asymmetric-ratio-edge} and $\alpha(G) \leq \Delta / \vec{\psi}(G)$ in \autoref{lem:asymmetric-ratio-vertex}.
Combining the two properties gives \autoref{thm:directed-edge-conductance} and \autoref{thm:directed-vertex-expansion}.

We highlight two new elements in the proofs of \eqref{e:refined},
one in dimension reduction and one in threshold rounding.
In the dimension reduction step, the Johnson-Lindenstrauss lemma can be used to project to a $1$-dimensional solution with a factor of $\log |V|$ loss as in~\cite{OZ22}.
For undirected vertex expansion, this was improved to a factor of $\log \Delta$ loss in two ways: one is the Gauassian projection method in~\cite{LRV13,KLT22}, while the other is a better analysis of dimension reduction for maximum matching in~\cite{JPV22}.
For directed edge conductance and directed vertex expansion, the SDP is more complicated and we do not know how to extend the Gaussian projection method to improve on the $\log |V|$ loss; see \autoref{sec:dimension-reduction-previous} for discussions.
Instead, we extend the approach in~\cite{JPV22} to prove that random projections only lose a factor of $\log \alpha(G)$ with high probability.
When the asymmetric ratio is small, we use Hoffman's result in \autoref{lem:Hoffman} about bounded-weighted circulations to prove a ``large optimal property'' of the SDPs (see \autoref{lem:large-optimal}), 
and use it to adapt the proof in~\cite{JPV22} for maximum weighted Eulerian subgraphs; see \autoref{sec:dimension-reduction} for details.

In the threshold rounding step of the $1$-dimensional solution,
we consider the dual SDP of $\vec{\lambda}_2^{v*}(G)$ and $\vec{\lambda}_2^{e*}(G)$ as in~\cite{KLT22}.
Unlike the dual SDP for undirected vertex expansion, 
these dual SDPs (see~\autoref{lem:l1-dual-vertex}) has some negative terms from some vertex potential function $r : V \to \R$.
The new idea in our threshold rounding is to not just consider the ordering defined by the vertex embedding function $f : V \to \R$ as usual,
but to consider the two orderings defined by $f \pm r$ and show that threshold rounding will work on one of these two orderings.
This idea also leads to a cleaner and nicer proof of the hard directions than that in~\cite{KLT22}, e.g.~without the preprocessing and postprocessing steps;
see \autoref{sec:rounding-algorithms} for details.

The generalizations of Cheeger inequalities for directed graphs and all Cheeger-type inequalities for hypergraphs are based on the same proofs of the corresponding results in~\cite{KLT22} with no new ideas involved.
We believe these results show that the reweighted eigenvalue approach provides a unifying method to lift the spectral theory for undirected edge conductance to obtain new results in more general settings in a systematic way.

Finally, we note that the maximum degree $\Delta$ for undirected vertex expansion, the asymmetric ratio $\alpha(G)$ for directed edge conductance and directed vertex expansion, and the maximum hyperedge size $r$ for hypergraph edge conductance all play the same role as a measure of how close the respective problem is to the basic problem of undirected edge conductance.
The trivial reductions to undirected edge conductance lose a factor of $\Delta$ for undirected vertex expansion, a factor of $\alpha(G)$ for directed edge-conductance (by just ignoring the directions), and a factor of $r$ for hypergraph edge conductance (by just considering the clique graph).
But the reductions through the reweighted eigenvalue approach only lose a factor of $\log \Delta$ in \eqref{e:vertex-Cheeger}, a factor of $\log \alpha(G)$ in \eqref{e:refined}, and a factor of $\log r$ in \autoref{thm:hypergraph-edge-conductance} respectively.

\subsection{Related Work} \label{sec:related-work}

There have been much interests in developing a spectral theory for directed graphs and hypergraphs, with many papers that we cannot review them all here.  
We describe the most relevant ones and compare to our work.

{\bf Nonlinear Laplacian for Directed Graphs}:
Yoshida~\cite{Yos16} introduced a nonlinear Laplacian operator for directed graphs and used it to define the following second eigenvalue
\[
\lambda_G = \inf_{x \perp \mu_G} \frac{\sum_{uv \in E} \Big( \big[ x_u / \sqrt{d_u} - x_v / \sqrt{d_v} \big]^+ \Big)^2}{\sum_{u \in V} x_u^2}
\]
where $\mu_G$ denotes the first eigenvector, $[a-b]^+$ denotes $\max\{a-b,0\}$, and $d_u$ is the total degree of $u$.
He considered the same directed edge conductance as in \autoref{def:directed-edge-conductance} and proved the Cheeger inequality that $\lambda_G / 2 \leq \vec{\phi}(G) \leq 2 \sqrt{\lambda_G}$, but did not give an approximation algorithm for computing $\lambda_G$ in~\cite{Yos16}.
Later, Yoshida~\cite{Yos19} gave an SDP approximation algorithm for computing $\lambda_{G}$ and this gives a polynomial time computable quantity $\tilde{\lambda}_G$ that satisfies
$\tilde{\lambda}_G \lesssim \vec{\phi}(G) \lesssim \sqrt{\tilde{\lambda}_G \log |V|}.$
We note that this is comparable but improved by our result\footnote{We remark that we can use the Johnson-Lindenstrauss lemma to do the dimension reduction step as in~\cite{OZ22}, and this would give $\vec{\phi}(G) \lesssim \sqrt{\vec{\lambda}_2^{e*}(G) \cdot \log|V|}$ as well.} for $\vec{\phi}(G)$ in \eqref{e:refined}, and cannot be used for certifying constant directed edge conductance as in \autoref{thm:directed-edge-conductance}.
We also note that this result is dominated by the SDP-based $O(\sqrt{\log|V|})$-approximation algorithm for $\vec{\phi}(G)$ in \cite{ACMM05} that we describe below.
To our knowledge, this is the only spectral formulation known in the literature that relates to directed edge conductance, and no spectral formulation was known for directed vertex expansion.
We also believe that the reweighted eigenvalue approach is simpler and more intuitive than the nonlinear Laplacian operator approach.

{\bf Approximation Algorithms Using Semidefinite Programming}:
In~\cite{ACMM05}, Agarwal, Charikar, Makarychev and Makarychev gave an SDP-based $O\big(\sqrt{\log|V|}\big)$-approximation algorithm for the directed sparsest cut problem on a directed graph $G=(V,E)$,
where the objective is to find a set $S$ that minimizes $|\delta^+(S)| / \min\{|S|,|\overline{S}|\}$.
We note that in the unweighted case, directed vertex expansion and directed edge conductance can be reduced to directed sparsest cut via standard reductions.
In the weighted case, the SDP for directed sparsest cut can be slightly modified to obtain a $O\big(\sqrt{\log |V|}\big)$-approximation algorithm for directed edge conductance; see \autoref{sec:ACMM}.
To our knowledge, it was not known that the SDP in~\cite{ACMM05} can be used to certify whether a directed graph has constant edge conductance as in \autoref{thm:directed-edge-conductance}, as the analysis using triangle inequalities based on~\cite{ARV09} has a $\sqrt{\log |V|}$ factor loss.
We show in \autoref{sec:ACMM} that the SDP in~\cite{ACMM05} is stronger than the SDP for directed edge conductance in \autoref{prop:lambda2-edge}.
Therefore, using the new analysis through asymmetric ratio in this paper, 
we also prove that the SDP in~\cite{ACMM05} provides a polynomial time computable quantity to certify constant directed edge conductance as in \autoref{thm:directed-edge-conductance}.

{\bf Cheeger Constant for Directed Graphs}:
Fill~\cite{Fil91} and Chung~\cite{Chu05} defined some symmetric matrices for directed graphs, and related their eigenvalues to Cheeger's constant and to mixing time.
Their formulations are similar to each other, but Chung's formulation is closer and more consistent with ours, as her work is also based on an Eulerian reweighted subgraph (which was called a circulation in~\cite{Chu05}) that we describe below.

Given a directed graph $G=(V,E)$ with a weight function $w : E \to \R_{\geq 0}$, let $P$ be the transition matrix of the ordinary random walks on $G$ with $P(u,v) = w(uv) / \sum_{v \in V} w(uv)$ for $uv \in E$.
Suppose $G$ is strongly connected, then there is a unique stationary distribution $\pi : V \to \R_{+}$ of the random walks on $G$ such that $\pi^T P = \pi^T$.
Let $\Pi := \diag(\pi)$.  Fill~\cite{Fil91} defined the product and the sum matrices as
$M(P) := P \Pi^{-1} P^T \Pi$ and
$A(P) := ( P + \Pi^{-1} P^T \Pi ) / 2$.
Chung~\cite{Chu05} noted that if the weight of an edge $uv$ is defined as $f(u,v) = \pi(u) \cdot P(u,v)$, then the weighted directed graph $G' = (V,E,f)$ is Eulerian such that $\sum_{u:uv \in E} f(u,v) = \sum_{w:vw \in E} f(v,w)$ for all $v \in V$.
Then she used the underlying weighted undirected graph to define the Laplacian of a directed graphs as
\begin{equation} \label{e:Chung-Laplacian}
\tilde{\L} = I - \big(\Pi^{1/2} P \Pi^{-1/2} + \Pi^{-1/2} P^T \Pi^{1/2}\big)/2
=  I - \Pi^{-\frac12} \Big(F+F^T\Big) \Pi^{-\frac12}/2
\vspace{-2mm}
\end{equation}
where $F = \Pi P$ is the adjacency matrix of $G'$.
Note that the spectrums of $A(P)$ and $\tilde{L}$ are essentially the same, as $P+\Pi^{-1}P^T\Pi$ and $\Pi^{1/2} P \Pi^{-1/2} + \Pi^{-1/2} P^T \Pi^{1/2}$ are similar matrices.
Note also that $\tilde{\L}$ is exactly the same as the normalized Laplacian matrix in the objective function in \autoref{def:directed-vertex-primal}.
The Cheeger constant of a directed graph~\cite{Fil91,Chu05} is defined as
\vspace{-1mm}
\begin{equation} \label{e:Cheeger-constant}
h(G) := \min_{S : S \neq \emptyset, S \neq V} h(S)
\quad {\rm where} \quad
h(S) 
= \frac{\sum_{u,v:u \in S, v \notin S} \pi(u) P(u,v)}{\min \{\pi(S), \pi(\overline{S})\}},
\vspace{-1mm}
\end{equation}

and Chung~\cite{Chu05} proved that 
$\lambda_2(\tilde{\L})/2 \leq h(G) \leq \sqrt{2\lambda_2(\tilde{\L})}$.

The main difference between our formulations and Chung's formulation is that we search for an {\em optimal} reweighting while Chung used a specific vertex-based reweighing by the stationary distribution.
We note that the Cheeger constant in \eqref{e:Cheeger-constant} could be very different from the directed edge conductance in \autoref{def:directed-edge-conductance} and the directed vertex expansion in \autoref{def:directed-vertex-expansion}; see \autoref{sec:Chung} for some examples.
We remark that many subsequent works used Cheeger constant as the objective for clustering and partitioning for directed graphs, and these examples illustrate their limitations in finding sets of small directed edge conductance or directed vertex expansion which are much more suitable notions for clustering and partitioning (see \cite{Yos16} for related discussions).

{\bf Mixing Time and Fastest Mixing Time:}
A main result in~\cite{Fil91,Chu05} is to use the second eigenvalue of $M(P), A(P), \lambda_2(\tilde{\L})$ to bound the mixing time of the ordinary random walks on $G$.
We state the result using Chung's formulation as it is closer to our formulation in \autoref{def:directed-vertex-expansion}.

\begin{theorem}[Bounding Mixing Time by Second Eigenvalue of Directed Graphs~\cite{Fil91,Chu05}] \label{thm:chung-mixing-time}
Let $G$ be a strongly connected directed graph $G=(V,E)$ with a weight function $w: E \to \R_{\geq 0}$, and $P$ be the transition matrix of the ordinary random walks on $G$ with $P(u,v) = w(uv) / \sum_{v \in V} w(uv)$ for $uv \in E$. 
Then the mixing time of the lazy random walks of $G$ with transition matrix $(I+P)/2$ to the stationary distribution $\pi$ is
\vspace{-2mm}
\[
\tau\Big( \frac{I+P}{2} \Big)
\lesssim \frac{1}{\lambda_2(\tilde{\L})} \cdot \log \Big( \frac{1}{\pi_{\min}} \Big)
\vspace{-2mm}
\]
where $\lambda_2(\tilde{\L})$ is the second smallest eigenvalue of the Laplacian in \eqref{e:Chung-Laplacian} and $\pi_{\min} = \min_{v \in V} \pi(v)$.
\end{theorem}
\vspace{-1mm}
We will use \autoref{thm:chung-mixing-time} to bound the fastest mixing time for general Markov chains in \autoref{thm:fastest-mixing}.
The fastest mixing time problem of reversible Markov chains was introduced by Boyd, Diaconis, and Xiao~\cite{BDX04}.  
This is a well-motivated problem in the study of Markov chains and has generated considerable interests (see the references in~\cite{OZ22}),
but there were no known combinatorial characterization of the fastest mixing time for quite some time.  
Recently, Oleskar-Taylor and Zanetti~\cite{OZ22} discovered a new Cheeger-type inequality relating reweighted second eigenvalue $\lambda_2^*$ and vertex expansion, and used it to give a combinatorial characterization of the fastest mixing time of reversible Markov chains by the vertex expansion of the graph. 
\autoref{thm:fastest-mixing} is a significant generalization of their result to general Markov chains, and we believe it is of independent interest.

{\bf Other Cheeger-Type Inequalities for Directed Graphs}:
Chan, Tang and Zhang~\cite{CTZ15} gave a higher-order Cheeger inequality for directed graphs.
Roughly speaking, they showed that there are $k$ disjoint subsets $S_1, \ldots, S_k \subseteq V$ with $\lambda_k(\tilde{\L}) \lesssim h(S_i) \lesssim k^2 \cdot \sqrt{\lambda_k(\tilde{\L})}$ for $1 \leq i \leq k$, where $h(S_i)$ is the Cheeger constant in \eqref{e:Cheeger-constant} and $\lambda_k(\tilde{\L})$ is the $k$-th smallest eigenvalue of the Laplacian in \eqref{e:Chung-Laplacian}.
The proof is a direct application of the higher-order Cheeger inequality for undirected graphs on the reweighted subgraph by the stationary distribution.
In \autoref{sec:Chung}, we show an example that rules out the possibility of having a higher-order Cheeger inequality for directed graphs relating $\lambda_k(\tilde{\L})$ to $k$-way directed edge conductance.

{\bf Other Hermitian Matrices of Directed Graphs}:
Besides the matrices in~\cite{Fil91,Chu05}, there are other Hermitian matrices associated to a directed graph studied in the literature.
Guo and Mohar~\cite{GM17} and Liu and Li~\cite{LL15} defined the Hermitian adjacency matrix $H$ of a directed graph as $H(u,v)=1$ if both $uv, vu \in E$, $H(u,v)=\imath$ if $uv \in E$ and $vu \notin E$ where $\imath$ is the imaginary number, $H(u,v)=-\imath$ if $uv \notin E$ and $vu \in E$, and $H(u,v)=0$ if both $uv, vu \notin E$.
There are also other Hermitian matrices defined for clustering directed graphs~\cite{LS20,CLSZ20} and for the Max-2-Lin problem~\cite{LSZ19}.
We confirm that there are no known relations between the eigenvalues of these Hermitian matrices and the expansion properties of a directed graph.

{\bf Directed Laplacian Solver Using Eulerian Reweighting}:
We note that the idea of reducing the problem for a directed graph to an Eulerian directed graph was also used in directed Laplacian solvers~\cite{CKP+16,CKP+17}.
As in~\cite{Chu05}, they also use the same reweighting by the stationary distribution to obtain an Eulerian graph from a directed graph.
(Furthermore, they introduced a notion of spectral sparsification of Eulerian directed graphs.)
We believe that the idea of reducing to Eulerian directed graphs and the concept of asymmetric ratio will find more applications in solving problems on directed graphs.

{\bf Spectral Theory for Hypergraphs}:
Louis~\cite{Lou15} and Chan, Louis, Tang, Zhang~\cite{CLTZ18} developed a spectral theory for hypergraphs.
They defined a continuous time diffusion process on a hypergraph $H=(V,E)$ and used it to define a {\em nonlinear} Laplacian operator and its eigenvalues $\gamma_1 \leq \gamma_2 \leq \ldots \leq \gamma_{|V|}$.
Then they derived a Cheeger inequality $\frac12 \gamma_2 \leq \phi(H) \leq \sqrt{2 \gamma_2}$,
where $\phi(H)$ is the hypergraph edge conductance of $H$ in \autoref{def:hypergraph-edge-conductance}.
But the quantity $\gamma_2$ is not polynomial time computable, 
and a semidefinite programming relaxation of $\gamma_2$ was used to output a set of edge conductance $O(\sqrt{\phi(H) \cdot \log r})$ where $r$ is the maximum size of a hyperedge.
They also proved an analog of higher-order Cheeger inequality for hypergraph edge conductance, such that for any $\epsilon \geq 1/k$ there are disjoint subsets $S_1, \ldots, S_{(1-\eps)k}$ with
\vspace{-2mm}
\begin{equation} \label{e:Louis-k}
\phi(S_i) \lesssim k^{2.5} \cdot \eps^{-1.5} \cdot \log k \cdot \log\log k \cdot \log r \cdot \sqrt{\tilde{\gamma}_k}
\end{equation}
for all $i \le (1-\eps)k$, where $\tilde{\gamma}_k$ can be thought of as a relaxation of $\gamma_k$.
They also gave an improved approximation algorithm for the small-set hypergraph edge conductance problem.

Using the reweighted eigenvalue approach, we can define the maximum reweighted $k$-th eigenvalue $\gamma_k^*$ as in \autoref{def:hypergraph-edge-primal}, and proved the following analog of higher-order Cheeger inequality for hypergraph edge conductance in \autoref{thm:higher-order-hypergraphs}:
For any $\eps \geq 1/k$, there are disjoint subsets $S_1, \ldots, S_{(1-\eps)k}$ with 
\vspace{-2mm}
\[
\phi(S_i) \lesssim \sqrt{k} \cdot \eps^{-4} \cdot \log{k} \cdot \sqrt{\log{r}} \cdot \sqrt{\gamma^*_k}
\]
for all $i \le (1-\eps)k$. This bound is comparable to that in~\cite{CLTZ18} when $\eps \approx 1/k$, and is an improvement when $\eps = \Theta(1)$ by a factor of more than $k^2$.
This also improves the approximation algorithm for the small-set hypergraph edge conductance problem in~\cite{CLTZ18} by a factor of more than $k$.
In addition, 
we also prove an analog of the improved Cheeger's inequality~\cite{KLLOT13} for hypergraphs.  
See \autoref{sec:hypergraphs} for the precise statements of all these results.

Compared to the spectral theory in~\cite{Lou15,CLTZ18} for hypergraphs using the continuous time diffusion process, we believe that the reweighted eigenvalue approach is simpler and more intuitive.
The definitions of the hypergraph diffusion process and its eigenvalues are quite technically involved and required considerable effort to make rigorous~\cite{CTWZ19}.
The reweighted eigenvalue approach allows us to recover and improve their results on hypergraph partitioning, and also to obtain a new result.
Since their spectral theory for hypergraph partitioning is gaining more attention in the machine learning community lately (e.g.~\cite{LM18}), we believe that it would be beneficial to have an alternative approach that is easier to understand and to prove new results and to have efficient implementations.

Finally, as a technical remark, we note that some careful reweighting schemes are crucially used in the construction of the diffusion process~\cite{Lou15,CLTZ18}, and also in recent exciting developments in hypergraph spectral sparsification~\cite{CKN20,KKTY21} (called balanced weight assignments).
This suggests that the concept of {\em reweighting} is central to these recent developments, and it would be very interesting to find connections between the different reweighting methods used in this work and these previous works.

\section{Preliminaries} \label{sec:prelim}

Given two functions $f, g$, we use $f \lesssim g$ to denote the existence of a positive constant $c > 0$, such that $f \le c \cdot g$ always holds.
We use $f \asymp g$ to denote $f \lesssim g$ and $g \lesssim f$.
For positive integers $k$, we use $[k]$ to denote the set $\{1, 2, \dots, k\}$.
For a function $f: X \rightarrow \R$, $\supp(f)$ denotes the domain subset on which $f$ is nonzero.
For an event $E$, $\one[E]$ denotes the indicator function that is $1$ when $E$ is true and $0$ otherwise.

\subsubsection*{Graphs and Hypergraphs}

The following definitions for hypergraphs also apply to undirected graphs.
Let $H = (V, E)$ be a hypergraph. A hyperedge $e \in E$ of $H$ is a nonempty subset of its vertices, i.e. $e \subseteq V$.
Throughout this paper, we use $n := |V|$ to denote the number of vertices and $m := |E|$ to denote the number of hyperedges.
Let $w : E \to \R_{\geq 0}$ be a weight function on the edges.
The weighted degree of a vertex $v$ is defined as $d_w(v) := \sum_{e:e \ni v} w(e)$.
The maximum degree of a graph is denoted as $\Delta := \max_{v \in V} d_{\vec{1}}(v)$, i.e.~the maximum unweighted degree of a vertex.
Let $S \subset V$ be a nonempty subset of vertices.
The edge boundary of $S$ is defined as $\delta(S) := \{e \in E \mid e \cap S \neq \emptyset {\rm~and~} e \cap \overline{S} \neq \emptyset\}$.
The volume of $S$ is defined as $\vol_w(S) := \sum_{v \in S} d_w(v)$.
The edge conductance of $S$ and of $G$ are defined as in \autoref{def:hypergraph-edge-conductance}.

The vertex boundary of $\emptyset \neq S \subset V$ is defined as $\partial(S) := \{v \in \overline{S} \mid \exists e \in \delta(S) \text{ with } v \in e\}$.
Let $\pi : V \to \R_{\geq 0}$ be a weight function on the vertices.
We write $\pi(S) := \sum_{v \in S} \pi(v)$ for a subset $S \subseteq V$.
In this paper, from \autoref{def:directed-vertex-expansion}, we define
$\psi(S) := \min\{ \pi(\partial(S)), \pi(\partial(\overline{S})) \} / \min\{\pi(S),\pi(\overline{S})\}$ and
$\psi(H) := \min_{\emptyset \neq S \subset V} \psi(S)$.
In~\cite{KLT22}, the weighted vertex expansion of $S$ and of $H$ are defined as 
$\psi(S) := \pi(\partial S) / \pi(S)$ and 
$\psi(H) := \min\{ 1, \min_{S \subseteq V:0 < \pi(S) \leq 1/2} \psi(S)\}$.
We remark that the two definitions of $\psi(H)$ can be shown to be within a factor of $2$ of each other.
The definition from \autoref{def:directed-vertex-expansion} has the advantages that $\psi(S) \leq 1$ for all $S \subseteq V$ and that it is more convenient in the threshold rounding proof in \autoref{sec:threshold-rounding}.

\subsubsection*{Directed Graphs}

Let $G = (V,E)$ be a directed graph.
Let $w : E \to \R_{\geq 0}$ be a weight function on the edges.
The weighted indegree of a vertex $v$ is defined as $d_w^-(v) := \sum_{u:uv \in E} w(uv)$
, the weighted outdegree of $v$ is defined as $d_w^+(v) := \sum_{u:vu \in E} w(vu)$, and the total weighted degree of $v$ is defined as $d_w(v) := d_w^+(v) + d_w^-(v)$. 
The maximum total degree of a graph is denoted as $\Delta := \max_{v \in V} d_{\vec 1}(v)$, i.e.~the maximum unweighted total degree of a vertex.
A directed graph is called Eulerian if $d_w^+(v) = d_w^-(v)$ for every $v \in V$.

Let $S \subset V$ be a nonempty subset of vertices.
The set of outgoing edges of $S$ is defined as $\delta^+(S) := \{uv \in E \mid u \in S {\rm~and~} v \notin S\}$, and the set of incoming edges to $S$ is defined as $\delta^-(S) := \delta^+(\overline{S})$.
The volume of $S$ is defined as $\vol_w(S) := \sum_{v \in S} d_w(v)$.
The directed edge conductance $\vec{\phi}(S)$ and $\vec{\phi}(G)$ are defined as in \autoref{def:directed-edge-conductance}.
The set of out-neighbors of $S$ is defined as $\partial^+(S) := \{v \notin S \mid \exists u \in S \text{ with } uv \in E\}$.
Let $\pi : V \to \R_{\geq 0}$ be a weight function on the vertices.
We write $\pi(S) := \sum_{v \in S} \pi(v)$ for a subset $S \subseteq V$.
The directed vertex expansion $\vec{\psi}(S)$ and $\vec{\psi}(G)$ are defined as in \autoref{def:directed-vertex-expansion}.

\subsubsection*{Random Walks}

Given a finite state space $X$, a Markov chain on $X$ is represented by a matrix $P \in \R^{X \times X}$, where $P(u, v)$ is the probability of traversing from state $u$ to state $v$ in one step. 
Thus, $P$ has nonnegative entries and satisfies $\sum_{v \in X} P(u, v) = 1$ for all $u \in X$.
Given an undirected or a directed graph $G=(V,E)$ with a weight function $w : E \to \R_{\geq 0}$, the transition matrix $P \in \R^{V \times V}$ of the random walks on $G$ is defined as $P(u,v) = w(uv) / \sum_{x \in V} w(ux)$.
A distribution $\pi: X \rightarrow \R$ is said to be a stationary distribution of $P$ if $\pi^T P = \pi^T$. 
It is well-known that if $G$ is strongly connected and aperiodic then there is a unique stationary distribution $\pi$.
Furthermore, if we do lazy random walks $P := (I+P')/2$ on a strongly connected graph, then $p_0 P^t \to \pi$ for any initial distribution $p_0$ when $t \to \infty$.
The mixing time measures how fast $p_0 P^t$ converges to $\pi$.
For a given Markov chain $P$ and $\eps \in (0, 1)$, the $\eps$-mixing time is defined as
\[
  \tau_{\eps}(P) := \min \Big\{t : \max_{p_0: V \rightarrow \R_{\ge 0}} \sum_{v \in V} \left| p_t(v) - \pi(v) \right| < \eps \Big\},
\]
where $p_0$ is an initial distribution on $V$ and $p_t$ denotes $p_0 P^t$.
The fastest mixing time $\tau^*(G)$ in \autoref{thm:fastest-mixing} is defined to be the minimum $(1/e)$-mixing time of a Markov chain on $V$, which is supported on $E$ and has stationary distribution $\pi$ as stated in \autoref{def:fastest-mixing}.

\subsubsection*{Spectral Graph Theory}

Given an undirected graph $G = (V, E)$ with a weight function $w : E \to R_{\geq 0}$, its adjacency matrix $A = A(G)$ is an $n \times n$ matrix where the $(u, v)$-th entry is $w(uv)$. 
The Laplacian matrix is defined as $L := D - A$, where $D:=\diag(\{d_w(v)\}_{v \in V})$ is the diagonal degree matrix. 
For any vector $x \in \R^n$, the Laplacian matrix has a useful quadratic form $x^T L x = \sum_{uv \in E} w(uv) \cdot (x(u)-x(v))^2$.

The normalized adjacency matrix is defined as $\mathcal{A} = D^{-1/2} A D^{-1/2}$, and the normalized Laplacian matrix is defined as $\mathcal{L} := I - \mathcal{A}$. 
Let $\lambda_1(\L) \leq \lambda_2(\L) \leq \cdots \leq \lambda_n(\L)$ be the eigenvalues of $\L$.
It is known that $\lambda_1(\L)=0$ with eigenvector $D^{1/2} \vec{1}$, and
\begin{equation} \label{e:lambda2}
\lambda_2(\L) = \min_{g \perp D^{1/2} \vec{1}} \frac{g^T \L g}{g^Tg}
= \min_{f \perp D \vec{1}} \frac{f^T L f}{f^T D f}
= \min_{f \perp D \vec{1}} \frac{\sum_{uv \in E} w(uv) \cdot (f(u)-f(v))^2}{\sum_v d_w(v) \cdot f(v)^2}.
\end{equation}

\subsubsection*{Semidefinite Programming}

Given a real symmetric matrix $M$, we say that $M$ is positive semidefinite (PSD) if $v^T M v \ge 0$ for all $v \in \R^n$, and we write $M \succeq 0$. 
Equivalently, $M \succeq 0$ if all its eigenvalues are nonnegative.
Also equivalently, $M \succeq 0$ if there exists a real matrix $X$ such that $M = X^T X$.
Let $x_i \in \R^n$ be the $i$-th column of $X$.
Then $M$ is called the Gram matrix of $x_1, \ldots, x_n \in \R^n$ as $M(i,j) = \inner{x_i}{x_j}$ for all $i,j \in [n]$.

Semidefinite programs (SDP) are optimization problems where the domain is $\{M \in \R^{n \times n} \mid M \succeq 0\}$, and the constraints and objective function are affine functions of the entries of $M$.
Equivalently, SDP's can be written as vector programs, where we optimize over the set of vectors $f_1,\ldots,f_n \in \R^n$, and the constraints and objective function are affine functions of the inner projects $\{\inner{f_i}{f_j}\}_{i,j \in [n]}$.

We will use the following exact SDP formulation of the second eigenvalue in \eqref{e:lambda2}:
\begin{equation} \label{e:lambda2-SDP}
\lambda_2(\L) 
= \min_{f: V \to \R^n,~\sum_{v} d_w(v) \cdot f(v) = \vec{0}} \frac{\sum_{uv \in E} w(uv) \cdot \norm{f(u)-f(v)}^2}{\sum_v d_w(v) \cdot \norm{f(v)}^2}.
\end{equation}

We will also use von Neumann's minimax theorem to construct semidefinite programming relaxations for reweighted second eigenvalues.
 
\begin{theorem}[Von Neumann's Minimax Theorem (see \cite{Sim95})] \label{thm:von-Neumann}
Let $X,Y$ be compact convex sets.
If $f$ is a real-valued continuous function on $X \times Y$ with 
$f(x,\cdot)$ concave on $Y$ for all $x \in X$
and $f(\cdot,y)$ convex on $X$ for all $y \in Y$,
then 
\[
\min_{x \in X} \max_{y \in Y} f(x,y) = \max_{y \in Y} \min_{x \in X} f(x,y).
\]
\end{theorem}

For reweighted $k$-th eigenvalues when $k > 2$,
we will use the following proposition in writing the maximum reweighted sum of $k$ smallest eigenvalue problem as a semidefinite program.

\begin{proposition} \label{prop:sum-of-lambda-k}
Let $X \in \R^{n \times n}$ be a symmetric matrix and let $1 \le k \le n$. 
Suppose the eigenvalues of $X$ are $\lambda_1 \le \lambda_2 \le \dots \le \lambda_n$. 
Then, $\lambda_1 + \lambda_2 + \cdots + \lambda_k$ is the value of the following SDP:
  \begin{eqnarray*}
    \min_{Y \in \R^{n \times n}} && \tr(XY) \\
    \st && 0 \preceq Y \preceq I_n \\
    && \tr(Y) = k,
  \end{eqnarray*}   
where $\tr(M) := \sum_{i=1}^n M_{i,i}$ is the trace of an $n \times n$ matrix $M$. 
\end{proposition}

\section{Cheeger Inequalities for Directed Graphs} \label{sec:main-results}

We prove the two main results \autoref{thm:directed-vertex-expansion} and \autoref{thm:directed-edge-conductance} in this section.
First, we prove the easy directions of the two results in \autoref{sec:easy}, and write the semidefinite programs for the reweighted eigenvalues in \autoref{sec:primal}.
Then, we show some properties of the asymmetric ratio in \autoref{sec:asymmetric-ratio}, and use these properties and the proof in~\cite{JPV22} to analyze a random projection algorithm to construct $1$-dimensional spectral solutions to the semidefinite programs in \autoref{sec:dimension-reduction}.
Then, we analyze a new threshold rounding algorithm for a $1$-dimensional solution to the dual programs, and prove the hard direction of the two results in \autoref{sec:rounding-algorithms}.
Finally, we show \autoref{thm:fastest-mixing} about fastest mixing time using~\cite{Fil91,Chu05} in \autoref{sec:fastest-mixing}, and provide details about the relations with some previous work that we mentioned in \autoref{sec:related-work} in \autoref{sec:relations}.

\subsection{Easy Directions by Reductions} \label{sec:easy}

There are two ways to prove the easy directions in \autoref{thm:directed-vertex-expansion} and \autoref{thm:directed-edge-conductance}.
A standard way is to construct a solution to $\vec{\lambda}_2^{v*}(G)$ or $\vec{\lambda}_2^{e*}(G)$ with small objective value when the directed vertex expansion or the directed edge conductance is small.
Instead, we use the reduction idea discussed in the introduction to prove the easy directions, as this is how we came up with the formulations and the reduction is the main theme in this paper.

\begin{proposition}[Easy Direction for Directed Vertex Expansion] \label{prop:directed-vertex-easy}

For any directed graph $G = (V,E)$ with weight function $\pi:V \to \R_{\geq 0}$, it holds that $\vec{\lambda}_2^{v*}(G) \leq 2\vec{\psi}(G)$.
\end{proposition}
\begin{proof}
The idea is to reduce directed vertex expansion of $G$ to the directed edge conductance of the reweighted Eulerian subgraph defined by $A$ in \autoref{def:directed-vertex-primal}, and then reduce to the underlying undirected graph defined by $\frac12 (A+A^T)$ and use classical Cheeger's inequality to lower bound its edge conductance by the second eigenvalue of its normalized Laplacian matrix.

Let $w(uv) := A(u,v)$ be the edge weight in the Eulerian reweighted subgraph for $uv \in E$.
For any nonempty $S \subset V$, by \autoref{def:directed-vertex-expansion} of directed vertex expansion and \autoref{def:directed-edge-conductance} of directed edge conductance,
\[
\vec{\psi}(S)
= \frac{\min\big\{ \pi(\partial^+(S)), \pi(\partial^+(\overline{S})) \big\}}{\min\{ \pi(S), \pi(\overline{S}) \}}
\geq \frac{2 \cdot \min\big\{ w(\delta^+(S)), w(\delta^-(S)) \big\}}{\min\{ \vol_w(S), \vol_w(\overline{S}) \}}
=2\vec{\phi}(S)
\]
where we use the degree constraints in \autoref{def:directed-vertex-primal} to establish that $w(\delta^+(S)) \leq \pi(\partial^+(S))$ and $w(\delta^-(S)) \leq \pi(\partial^+(\overline{S}))$ (note that they are not necessarily equalities because of the self-loops), and $\vol_w(S) = 2\pi(S)$ for every nonempty $S \subset V$.

As the edge-weighted directed graph $G'=(V,E,w)$ is Eulerian, 
it holds that $w(\delta^+(S)) = w(\delta^-(S))$ for every nonempty $S \subset V$,
and thus the directed edge conductance of $G'$ is equal to half the edge conductance of the underlying undirected graph $G''$ with edge weight $w''(uv) = \frac{1}{2} \big(w(uv) + w(vu)\big)$, because
\[
2 \vec{\phi}(S) 
= \frac{\min\big\{ w(\delta^+(S)), w(\delta^-(S)) \big\}}{\frac{1}{2} \cdot \min\{ \vol_w(S), \vol_w(\overline{S}) \}}
= \frac{w''(\delta(S))}{ \min\{ \vol_{w''}(S), \vol_{w''}(\overline{S}) \}}
= \phi(S).
\]
As the graph $G''$ is undirected, 
we can use Cheeger's inequality in \eqref{e:Cheeger} to lower bound the edge conductance of $G''$ by the second smallest eigenvalue of its normalized Laplacian matrix $\L(A) := I - \frac12 \Pi^{-1/2} (A+A^T) \Pi^{-1/2}$.
Therefore, for any nonempty $S \subset V$,
$\vec{\psi}(S) \geq 2 \vec{\phi}(S) = \phi(S) \geq \lambda_2(\L(A)) / 2$.
Since this holds for any nonempty $S \subset V$ and any weighted Eulerian subgraph defined by $A$ satisfying the constraints in \autoref{def:directed-vertex-primal},
we conclude that $2 \psi(G) \geq \max_A \lambda_2(\L(A)) = \vec{\lambda}_2^{v*}(G)$.
\end{proof}

The proof of the easy direction of \autoref{thm:directed-edge-conductance} is similar, but with a subtle difference in handling the denominator.

\begin{proposition}[Easy Direction for Directed Edge Conductance] \label{prop:directed-edge-easy}
For any directed graph $G = (V,E)$ with weight function $w:E \to \R_{\geq 0}$, it holds that $\vec{\lambda}_2^{e*}(G) \leq 2\vec{\phi}(G)$.
\end{proposition}

\begin{proof}
Let $w'(uv) := A(u,v)$ be the edge weight in the Eulerian reweighted subgraph $G'$ in \autoref{def:directed-edge-primal}.
Let $G''$ be the underlying undirected graph with edge weight $w''(uv) := \frac12 (w'(uv)+w'(vu))$, with an additional self-loop on each vertex so that the weighted degree $d_{w'}(v)$ on each vertex $v$ is exactly equal to the total degree $d_w(v) = \sum_{u \in V} (w(uv)+w(vu))$ of $v$ in $G$.
Then, by the edge capacity constraints and the Eulerian constraints in \autoref{def:directed-edge-primal}, for any nonempty $S \subset V$,
\[  \vec{\phi}(S) 
   = \frac{\min\big\{ w(\delta^+(S)), w(\delta^-(S)) \big\}}{\min\{ \vol_w(S), \vol_w(\overline{S}) \}}
   \geq
   \frac{\min\big\{ w'(\delta^+(S)), w'(\delta^-(S)) \big\}}{\min\{ \vol_w(S), \vol_w(\overline{S}) \}}
   =
   \frac{w''(\delta(S))}{\min\{ \vol_{w''}(S), \vol_{w''}(\overline{S}) \}}
   = \phi(S).
\]

Let $D := \diag(d_w)$ be the diagonal degree matrix of $G''$,
and $\L := D^{-1/2}(D - \frac12 (A+A^T))D^{-1/2} = I - \frac12 D^{-1/2} (A+A^T) D^{-1/2}$ be the normalized Laplacian matrix of $G''$.
As $G''$ is undirected, it follows from Cheeger's inequality in \eqref{e:Cheeger} that $\phi(S) \geq \lambda_2(\L) / 2$.
Since this holds for any nonempty $S \subset V$ and any weighted subgraph defined by $A$ satisfying the constraints in \autoref{def:directed-edge-primal},
we conclude that $\vec{\phi}(G) \geq \phi(G'') \geq \max_A \lambda_2(\L) / 2 = \vec{\lambda}_2^{e*}(G) / 2$.
\end{proof}

\subsection{Semidefinite Programs} \label{sec:primal}

We show that the optimization problems of reweighted eigenvalues can be formulated exactly as semidefinite programs, and so they can be approximated arbitrarily well in polynomial time.
The construction is similar to that of the semidefinite program for undirected vertex expansion in~\cite{BDX04,Roc05}, but von Neumann minimax theorem is used instead of SDP duality.

\begin{proposition}[SDP for Reweighted Second Eigenvalue with Vertex Capacity Constraints]
\label{prop:lambda2-vertex}
Given a directed graph $G = (V, E)$ and a weight function $\pi : V \to \R_{\geq 0}$, the optimization problem in \autoref{def:directed-vertex-primal} can be written as
  \begin{align*}
    \vec{\lambda}_2^{v*}(G) :=
     \min_{f: V \rightarrow \R^n} \max_{A \geq 0} &~~~ \frac12 \sum_{uv \in E} A(u,v) \cdot \norm{f(u) - f(v)}^2
    \\
    \st&~~~
    A(u, v) = 0 & & \forall uv \not \in E
    \\
    &~~~
    \sum_{v \in V} A(u, v) = \sum_{v \in V} A(v, u)  & & \forall u \in V
    \\
    &~~~
    \sum_{v \in V} A(v, u) = \pi(u) & & \forall u \in V
    \\
    &~~~
    \sum_{v \in V} \pi(v) \cdot f(v) = \vec{0}
    \\
    &~~~ \sum_{v \in V} \pi(v) \cdot \norm{f(v)}^2 = 1.
  \end{align*}
\end{proposition}

\begin{proof}
Let $\L := I - \frac12 \Pi^{-1/2} (A+A^T) \Pi^{-1/2}$ be the normalized Laplacian matrix in the objective function $\max_A \lambda_2(\L)$ in \autoref{def:directed-vertex-primal}.
By \eqref{e:lambda2},
\[
\lambda_2(\L)
= \min_{f \perp \Pi \vec{1}} \frac{\sum_{(u,v) \in \binom{V}{2}} \frac12 \big(A(u,v) + A(v,u)\big) \cdot |f(u)-f(v)|^2}{\sum_v \pi(v) f(v)^2}.
\]
Then we write $f \perp \Pi \vec{1}$ as the second last constraint and normalize the denominator to $1$ as the last constraint.
By \eqref{e:lambda2-SDP}, the SDP relaxation where we replace $f : V \to \R$ by $f : V \to \R^n$ is an exact relaxation. 
After the SDP relaxation, the feasible domain becomes convex and so we can apply von Neumann minimax theorem in \autoref{thm:von-Neumann} to switch the order of $\max_A \min_f$ in \autoref{def:directed-vertex-primal} to $\min_f \max_A$ as in the statement of this lemma.
\end{proof}

The same construction is used for $\vec{\lambda}_2^{e*}(G)$ in \autoref{def:directed-edge-primal} and the proof is omitted.

\begin{proposition}[SDP for Reweighted Second Eigenvalue with Edge Capacity Constraints]
\label{prop:lambda2-edge}
Given a directed graph $G = (V, E)$ and a weight function $w : E \to \R_{\geq 0}$, the optimization problem in \autoref{def:directed-edge-primal} can be written as
  \begin{align*}
    \vec{\lambda}_2^{e*}(G) :=
     \min_{f: V \rightarrow \R^n} \max_{A \geq 0} &~~~ \frac12 \sum_{uv \in E} A(u, v) \cdot \norm{f(u) - f(v)}^2
    \\
    \st&~~~
    A(u, v) = 0 & & \forall uv \not \in E
    \\
    &~~~
    \sum_{v \in V} A(u, v) = \sum_{v \in V} A(v, u)  & & \forall u \in V
    \\
    &~~~
    A(u, u) \leq w(uv)  & & \forall uv \in E
    \\
    &~~~
    \sum_{v \in V} d_w(v) \cdot f(v) = \vec{0}
    \\
    &~~~ \sum_{v \in V} d_w(v) \cdot \norm{f(v)}^2 = 1.
  \end{align*}
\end{proposition}

We will use these semidefinite programs to prove the two main results.

\subsection{Asymmetric Ratio} \label{sec:asymmetric-ratio}

A key parameter in our proofs is the asymmetric ratio $\alpha(G)$ in \autoref{def:asymmetric-ratio}.
This parameter satisfies two useful properties.
One is that $\alpha(G)$ can be used to bound the directed edge conductance and directed vertex expansion.
Another is that directed graphs with bounded asymmetric ratio satisfy the ``large optimal property'' that we will describe in \autoref{sec:large-optimal}, which can be used in the proof in~\cite{JPV22} to provide a better analysis of the random projection algorithm for dimension reduction of the SDP solutions.

\subsubsection{Asymmetric Ratio and Expansion Properties}

The relation between asymmetric ratio of edge-weighted graph and directed edge conductance is simple.

\begin{lemma}[Asymmetric Ratio and Directed Edge Conductance] \label{lem:asymmetric-ratio-edge}
For any directed graph $G=(V,E)$ and any weight function $w : E \to \R_{\geq 0}$, it holds that $\alpha(G)\leq 1/\vec{\phi}(G)$.
\end{lemma}
\begin{proof}
Let $S \subset V$ be a nonempty set. Suppose $\vol_w(S) \leq \vol_w(\overline{S})$; the other case is similar.
Then, by the definition of directed edge conductance in \autoref{def:directed-edge-conductance}, 
$w(\delta^+(S)) \geq \vec{\phi}(G) \cdot \vol_w(S)$ and $w(\delta^-(S)) \geq \vec{\phi}(G) \cdot \vol_w(S)$.
On the other hand, $w(\delta^+(S)) \leq \vol_w(S)$ and $w(\delta^-(S)) \leq \vol_w(S)$. 
Therefore, $\alpha(S) = w(\delta^+(S)) / w(\delta^-(S)) \leq 1/\vec{\phi}(G)$ for any nonempty $S \subset V$, 
and we conclude that $\alpha(G) \leq 1/\vec{\phi}(G)$.
\end{proof}

The relation between asymmetric ratio of vertex-weighted graph and directed vertex expansion is less trivial and has a dependency on the maximum total degree $\Delta$.

\begin{lemma}[Asymmetric Ratio and Directed Vertex Expansion] \label{lem:asymmetric-ratio-vertex}
For any directed graph $G=(V,E)$ and any weight function $\pi : V \to \R_{\geq 0}$, it holds that $\alpha(G) \lesssim \Delta/\vec{\psi}(G)$.
\end{lemma}

\begin{proof}
To upper bound $\alpha(G)$ for a vertex-weighted graph, by \autoref{def:asymmetric-ratio}, we need to upper bound $\alpha(S) = w_{\pi}(\delta^+(S)) / w_{\pi}(\delta^+(\overline{S}))$ and $\alpha(\overline{S}) = w_{\pi}(\delta^+(\overline{S})) / w_{\pi}(\delta^+(S))$ for any nonempty $S \subset V$, where $w_{\pi}(uv) = \min\{\pi(u),\pi(v)\}$ is the $\pi$-induced edge weight for $uv \in E$.
We assume without loss of generality that $\pi(S) \leq \pi(V)/2$.

For the numerators, note that
$
w_{\pi}(\delta^+(S)) \leq \sum_{u \in S} \sum_{v: uv \in E} w_{\pi}(uv) \leq \sum_{u \in S} \Delta \cdot \pi(u) = \Delta \cdot \pi(S),
$
and similarly
$
w_{\pi}(\delta^+(\overline{S})) = w_{\pi}(\delta^-(S)) \leq \sum_{u \in S} \sum_{v: vu \in E} w_{\pi}(vu) \leq \sum_{u \in S} \Delta \cdot \pi(u) = \Delta \cdot \pi(S).
$
Therefore, the numerators for $\alpha(S)$ and $\alpha(\overline{S})$ are at most $\Delta \cdot \pi(S)$. 
For the denominators, we claim that $w_{\pi}(\delta^+(S)) \geq \frac{1}{3} \vec{\psi}(G) \cdot \pi(S)$ and $w_{\pi}(\delta^+(\overline{S})) \geq \frac{1}{3} \vec{\psi}(G) \cdot \pi(S)$.
This claim implies that the denominators for $\alpha(S)$ and $\alpha(\overline{S})$ are at least $\frac{1}{3} \vec{\psi}(G) \cdot \pi(S)$, and the lemma follows immediately.

To prove the claim, we first consider the lower bound on $w_{\pi}(\delta^+(S))$. 
Let $\eps := \vec{\psi}(G) / 3$.
Suppose by contradiction that $w_\pi(\delta^+(S)) < \epsilon \cdot \pi(S)$.  
Let $C_S := \{u \in S \mid \exists v {\rm~with~} uv \in \delta^+(S) {\rm~and~} \pi(u) \leq \pi(v)\}$
and $C_{\overline{S}} := \{v \in \overline{S} \mid \exists u {\rm~with~} uv \in \delta^+(S) {\rm~and~} \pi(v) \leq \pi(u)\}$.
Since each $u \in C_S$ contributes at least $\pi(u)$ weight to $w_{\pi}(\delta^+(S))$ and these contributions are disjoint, it follows that $\pi(C_S) \leq w_{\pi}(\delta^+(S)) < \eps \cdot \pi(S)$.
By the same argument, $\pi(C_{\overline{S}}) < \eps \cdot \pi(S)$.
Note that, by definition of $C_S$ and $C_{\overline{S}}$, each edge in $\delta^+(S)$ has at least one vertex in $C_S \cup C_{\overline{S}}$.
This implies that $\partial^+(S-C_S) \subseteq C_S \cup C_{\overline{S}}$,
but this leads to the contradiction that 
\[
\vec{\psi}(S-C_S) 
\leq \frac{\pi(\partial^+(S-C_S))}{\pi(S-C_S)}
\leq \frac{\pi(C_S \cup C_{\overline{S}})}{\pi(S) - \pi(C_S)}
< \frac{2\eps \cdot \pi(S)}{(1-\eps) \cdot \pi(S)}
= \frac{2\vec{\psi}(G)}{3(1-\vec{\psi}(G)/3)}
\leq \vec{\psi}(G).
\]
The lower bound on $w_{\pi}(\delta^+(\overline{S}))$ is by a similar argument.
Suppose by contradiction that $w_\pi(\delta^+(\overline{S})) < \epsilon \cdot \pi(S)$.  
Let $C_{\overline{S}} := \{u \in \overline{S} \mid \exists v {\rm~with~} uv \in \delta^+(\overline{S}) {\rm~and~} \pi(u) \leq \pi(v)\}$
and $C_{S} := \{v \in S \mid \exists u {\rm~with~} uv \in \delta^+(\overline{S}) {\rm~and~} \pi(v) \leq \pi(u)\}$.
Once again, it follows that $\pi(C_S), \pi(C_{\overline{S}}) \leq w_{\pi}(\delta^+(\overline{S})) < \eps \cdot \pi(S)$, and $\partial^+(\overline{S} - C_{\overline{S}}) \subseteq C_S \cup C_{\overline{S}}$.
But this leads to the contradiction that
\[
\vec{\psi}(\overline{S}-C_{\overline{S}}) 
= \frac{\pi(\partial^+(\overline{S}-C_{\overline{S}}))}{ \min\{ \pi(\overline{S}-C_{\overline{S}}), \pi(S+C_{\overline{S}})  \} }
\leq \frac{\pi(C_S \cup C_{\overline{S}})}{ \min\{ \pi(\overline{S}) - \eps \cdot \pi(S), \pi(S) \}}
< \frac{2\eps \cdot \pi(S)}{(1-\eps)\cdot \pi(S)}
\leq \vec{\psi}(G),
\]
where the second inequality uses that $\pi(\overline{S}) \geq \pi(S)$.
This completes the proof of the claim.
\end{proof}

\subsubsection{Asymmetric Ratio and Large Optimal Property} \label{sec:large-optimal}

Consider the semidefinite programs for $\vec{\lambda}_2^{v*}(G)$ and $\vec{\lambda}_2^{e*}(G)$ in \autoref{prop:lambda2-vertex} and \autoref{prop:lambda2-edge}.
When the geometric embedding $f : V \to \R^n$ in the outer minimization problem is fixed, the inner maximization problem is simply to find a maximum weighted Eulerian subgraph $A$ with vertex capacity constraints in \autoref{prop:lambda2-vertex} and with edge capacity constraints in \autoref{prop:lambda2-edge}.
The following are trivial upper bounds on the optimal values of the inner maximization problems.

\begin{claim}[Maximum Weighted Eulerian Subgraph with Capacity Constraints] \label{def:Eulerian}
Given a directed graph $G=(V,E)$ and an embedding $f : V \to \R^n$, 
let $\nu_f^{v*}(G)$ and $\nu_f^{e*}(G)$ be the objective values of the inner maximization problem in \autoref{prop:lambda2-vertex} and \autoref{prop:lambda2-edge} respectively.
Then
\[
\nu_f^{v*}(G) \leq \frac12 \sum_{uv \in E} w_{\pi}(uv) \cdot \norm{f(u)-f(v)}^2
\quad {\rm and} \quad
\nu_f^{e*}(G) \leq \frac12 \sum_{uv \in E} w(uv) \cdot \norm{f(u)-f(v)}^2,
\]
where $w_{\pi}(uv) = \min\{\pi(u),\pi(v)\}$ be the $\pi$-induced edge weight function defined in \autoref{def:asymmetric-ratio}.
\end{claim}

In the undirected vertex expansion problem~\cite{OZ22,JPV22}, when $\pi(v)=1$ for all $v\in V$, the inner maximization problem is exactly the maximum weighted fractional matching problem.
Jain, Pham and Vuong~\cite{JPV22} used the fact that any graph with maximum degree $\Delta$ has an edge coloring with at most $\Delta+1$ colors to show that the inner maximization problem has a solution with weight at least $1/(\Delta+1)$ fraction of the trival upper bound.
They then used this ``large optimal property'' to analyze a dimension reduction algorithm for maximum weighted matching; see \autoref{sec:dimension-reduction}.

We observe that the asymmetric ratio $\alpha(G)$ in \autoref{def:asymmetric-ratio} can be used to play the same role as $\Delta$ to establish the large optimal property for the maximum weighted Eulerian subgraph problems in \autoref{def:Eulerian}.
The proof uses the following characterization of asymmetric ratio by Hoffman (see also \cite[Theorem 2.3]{EMPS16}), rephrased using our terminologies.

\begin{lemma}[Hoffman's Circulation Lemma] \label{lem:Hoffman}
Let $G=(V,E)$ be a directed graph with a weight function $w : E \to R_{\geq 0}$.
Then $G$ has asymmetric ratio at most $\alpha$ if and only if there exists an Eulerian reweighting $A$ of $G$ such that 
\[
\sum_{v: uv \in E} A(u, v) = \sum_{v: vu \in E} A(v, u) {\rm~~for~all~} u \in V
\quad {\rm and} \quad
w(uv) \leq A(u,v) \leq \alpha \cdot w(uv) {\rm~~for~all~} uv \in E.
\]
\end{lemma}

The large optimal property in terms of asymmetric ratio is a simple consequence of Hoffman's circulation lemma.

\begin{lemma}[Large Optimal Property] \label{lem:large-optimal}
Given a directed graph $G=(V,E)$ and an embedding $f : V \to \R^n$, 
let $\nu_f^{v*}(G)$ and $\nu_f^{e*}(G)$ be the objective values of the inner maximization problem in \autoref{prop:lambda2-vertex} and \autoref{prop:lambda2-edge} respectively.
Then
\[
\nu_f^{v*}(G) \geq \frac{1}{2\Delta \cdot \alpha(G)} \sum_{uv \in E} w_{\pi}(uv) \cdot \norm{f(u)-f(v)}^2
{\rm~and~}
\nu_f^{e*}(G) \geq \frac{1}{2\alpha(G)} \sum_{uv \in E} w(uv) \cdot \norm{f(u)-f(v)}^2.
\]
\end{lemma}
\begin{proof}
First, consider $\nu_f^{e*}(G)$ in \autoref{prop:lambda2-edge} with weight function $w : E \to \R_{\geq 0}$.
Let $A$ be an Eulerian reweighting of $G$ with weight function $w$ given in \autoref{lem:Hoffman}.
As $w(uv) \leq A(u,v) \leq \alpha(G) \cdot w(uv)$ for $uv \in E$,
the scaled-down subgraph $A/\alpha(G)$ satisfies the edge capacity constraints 
and is a feasible solution to the inner maximization problem in \autoref{prop:lambda2-edge}, with objective value $\frac{1}{2} \sum_{uv \in E} \frac{A(u,v)}{\alpha(G)} \cdot \norm{f(u)-f(v)}^2 \geq \frac{1}{2 \alpha(G)} \sum_{uv \in E} w(uv) \cdot \norm{f(u)-f(v)}^2$.

Similarly, consider $\nu_f^{v*}(G)$ in \autoref{prop:lambda2-vertex} with weight function $\pi : V \to \R_{\geq 0}$ and induced function $w_{\pi} : E \to \R_{\geq 0}$.
Let $A$ be an Eulerian reweighting of $G$ with weight function $w_{\pi}$ given in \autoref{lem:Hoffman}.
For each vertex $u$, the weighted degree is $\sum_{v:uv \in E} A(u,v) \leq \alpha(G) \cdot \sum_{v:uv \in E} w_{\pi}(uv) \leq \alpha(G) \cdot \Delta \cdot \pi(u)$.
Therefore, scaling down $A$ by a factor of $\Delta \cdot \alpha(G)$ satisfies the vertex capacity constraints and is a feasible solution to the inner maximization problem of \autoref{prop:lambda2-vertex}, with objective value $\frac{1}{2\Delta \cdot \alpha(G)} \sum_{uv \in E} w_{\pi}(uv) \cdot \norm{f(u)-f(v)}^2$. 
\end{proof}

We will use \autoref{lem:large-optimal} in the analysis of the dimension reduction step in the next subsection.

\subsection{Dimension Reduction} \label{sec:dimension-reduction}

The goal in this subsection is to obtain a good low-dimensional solution to the semidefinite programs in \autoref{prop:lambda2-vertex} and \autoref{prop:lambda2-edge}.

\begin{definition}[Low-Dimensional Solutions to Semidefinite Programs] \label{def:k-dim-primal}
Define 
\[
\vec{\lambda}_v^{(k)}(G) := \min_{f : V \to \R^k} \max_{A \geq 0} \frac12 \sum_{uv \in E} A(u,v) \cdot \norm{f(u)-f(v)}^2
\] 
to be the objective value of the SDP in \autoref{prop:lambda2-vertex} when restricting $f$ to be a $k$-dimensional embedding and subjecting to the same constraints.

Define $\vec{\lambda}_e^{(k)}(G)$ similarly as the objective value of the SDP in \autoref{prop:lambda2-edge} when restricting $f$ to be a $k$-dimensional embedding subjecting to the same constraints.
\end{definition}

The main result that we will prove in this subsection is that there is a good $1$-dimensional solution when the asymmetric ratio of the graph is small.

\begin{theorem}[One Dimensional Solutions to Semidefinite Programs] \label{thm:dimension-reduction}
Let $\vec{\lambda}_v^{(k)}(G)$ and $\vec{\lambda}_e^{(k)}(G)$ be as defined in \autoref{def:k-dim-primal}.
Then 
\[
\vec{\lambda}_v^{(1)}(G) \lesssim \log \big(\Delta \cdot \alpha(G)\big) \cdot \vec{\lambda}_2^{v*}(G)
\quad {\rm and} \quad
\vec{\lambda}_e^{(1)}(G) \lesssim \log \alpha(G) \cdot \vec{\lambda}_2^{e*}(G).
\]
\end{theorem}

\begin{remark} \label{r:tight}
Using the tight example in~\cite{KLT22} for undirected vertex expansion and the standard reduction from undirected vertex expansion to directed edge conductance, we can show that the second inequality in \autoref{thm:dimension-reduction} is tight up to a constant factor.
However, as this example has large maximum degree, we cannot conclude that the first inequality in \autoref{thm:dimension-reduction} is also tight.
\end{remark}

\subsubsection{Previous Work} \label{sec:dimension-reduction-previous}

For undirected vertex expansion, there are two different proofs~\cite{KLT22,JPV22} of an analogous dimension reduction result with a factor of $\log \Delta$ loss.

In~\cite{KLT22}, the approach was to first construct the dual SDP of $\lambda_2^*$, where the objective function is of the form $\min_{f:V \to \R^n} \sum_{v \in V} \pi(v) \cdot \max_{u:uv \in E} \norm{f(u)-f(v)}^2$.
Since each maximum is over at most $\Delta$ terms,
one can use the analysis of the Gaussian projection method in~\cite{LRV13} to directly project $f$ to a $1$-dimensional solution,
and prove that the expected maximum is at most a factor of $\log \Delta$ larger using properties of Gaussian random variables.

For the semidefinite programs for $\vec{\lambda}_2^{v*}(G)$ and $\vec{\lambda}_2^{e*}(G)$, however, the objective function of the dual SDP is of the form $\min_{f:V \to \R^n} \sum_{v \in V} \pi(v) \cdot \max_{u:uv \in E} \big( \norm{f(u)-f(v)}^2 - r(u) + r(v) \big)$ where $r(u)$ is a real number (see \autoref{lem:l1-dual-vertex} for the $1$-dimensional version).
Since the contribution of $-r(u)+r(v)$ could be negative, the same approach of projecting $f$ does not work anymore.
(We have an example showing that the random projection algorithm in~\cite{LRV13,KLT22} will lose a factor of $\log \alpha(G)$, even when the maximum degree is constant.)

Instead, we will follow the two-step approach of projecting the (primal) SDP solution in~\cite{JPV22}.
In the first step, the $n$-dimensional solution to $\lambda_2^{*}$ is projected to a $O(\log \Delta)$-dimensional solution, while the objective value only increases by a constant factor.
Then the $O(\log \Delta)$-dimension solution is reduced to a $1$-dimension solution, by choosing the best coordinate and losing a factor of $O(\log \Delta)$ as in \cite{OZ22}.

To analyze the first step, they proved a dimension reduction theorem for maximum weighted matchings.
We observe that their proof only needs the large optimal property of maximum matching as discussed in \autoref{sec:large-optimal} (but not any other property specific to matchings), and so it also works for maximum weighted Eulerian subgraphs in our problems with \autoref{lem:large-optimal} about their large optimal property in place.

\subsubsection{Random Projection}

We remark that the arguments in this subsubsection are essentially the same as in \cite{JPV22}. 
We cannot directly use their theorem as a black box and so we reproduce their arguments here.

\begin{definition}[Random Projection Algorithm] \label{def:random-projection-algorithm}
Let $G=(V,E)$ be a directed graph and $f: V \to \R^n$ be an embedding of the vertices in $G$.
For $1 \leq i \leq k$, let $g_i \sim \mathcal{N}(0,I_n)$ be an i.i.d. Gaussian vector.
Define $\Gamma : \R^n \to \R^k$ to be the random Gaussian projection operator with 
\[
\Gamma f(v) = \frac{1}{\sqrt{k}} \big( \ip{f(v)}{g_i} \big)_{i=1}^k.
\]
\end{definition}

The following properties of the random projection algorithm will be used.

\begin{lemma}[Gaussian Properties~\cite{MMR19,JPV22}] \label{lem:Gaussian}
Let $G=(V,E)$ be a directed graph and $f: V \to \R^n$ be an embedding of the vertices in $G$.
Let $\Gamma: \R^n \to \R^k$ be the random Guassian projection operator in \autoref{def:random-projection-algorithm} and let $h : V \to \R^k$ be the random projected solution with $h(v) := \Gamma f(v)$ for $v \in V$.
There exists a constant $c$ that satisfies the following two properties.
For all $u,v \in V$,
\[
\Pr_h \big[ \norm{h(u) - h(v)} \notin e^{\pm \eps} \norm{f(u)-f(v)} \big] \leq e^{-c \eps^2 k}.
\]
For all $u,v \in V$, let ${\mathcal{E}}_{u,v}$ be the event that $\norm{h(u) - h(v)} \geq e^{\eps} \norm{f(u)-f(v)}$, then
\[
\E_h \bigg[ \one_{\mathcal{E}_{u,v}} \bigg( \frac{\norm{h(u)-h(v)}^2}{\norm{f(u)-f(v)}^2} - e^{2\eps} \bigg) \bigg] \leq e^{-c\eps^2 k}. 
\]
\end{lemma}

The main technical result is the following adaptation of the dimension reduction theorem for maximum matchings in~\cite{JPV22}.

\begin{theorem}[Dimension Reduction for Maximum Weighted Eulerian Subgraphs] \label{thm:JPV}
Let $\vec{\lambda}_v^{(k)}(G)$ and $\vec{\lambda}_e^{(k)}(G)$ be as defined in \autoref{def:k-dim-primal}.
There exists a constant $C$ such that
\[
\vec{\lambda}_v^{\big( C \cdot \log (\Delta \cdot \alpha(G)) \big)}(G) \lesssim \vec{\lambda}_2^{v*}(G)
\quad {\rm and} \quad
\vec{\lambda}_e^{\big( C \cdot \log \alpha(G) \big)}(G) \lesssim \vec{\lambda}_2^{e*}(G).
\]
\end{theorem}

\begin{proof}
The proofs of the two inequalities are essentially the same,
and we explain the proof of the second inequality here.
Let $G=(V,E)$ be a directed graph and $f: V \to \R^n$ be an optimal embedding of the vertices in $G$ such that 
$\nu_f^{e*}(G) = \vec{\lambda}_2^{e*}(G)$.
Let $h : V \to \R^k$ be the random projected solution with $h(v) := \Gamma f(v)$ for $v \in V$.
We would like to use $h$ as a solution to $\vec{\lambda}_e^{(k)}(G)$.
First, note that 
\[
\sum_{v \in V} d(v) \cdot h(v)
= \sum_{v \in V} d(v) \cdot \Gamma f(v)
= \Gamma \Big( \sum_{v \in V} d(v) \cdot f(v) \Big) 
= 0,
\]
and so $h$ also satisifes this constraint in the SDP in \autoref{prop:lambda2-edge}.
But the normalization constraint $\sum_{v \in V} d(v) \cdot \norm{h(v)}^2 = 1$ may not be satisfied, and the objective value $\nu_h^{e*}(G) = \max_{A \geq 0} \frac12 \sum_{uv \in E} A(u,v) \cdot \norm{h(u)-h(v)}^2$ may be bigger than $\nu_f^{e*}(G)$.
Our plan is to prove that 
\begin{equation} \label{e:plan}
\vec{\lambda}_e^{(k)}(G) \leq \frac{\nu_h^{e*}(G)}{\sum_{v \in V} d(v) \norm{h(v)}^2}
\lesssim \frac{\nu_f^{e*}(G)}{\sum_{v \in V} d(v) \norm{f(v)}^2} 
= \vec{\lambda}_2^{e*}(G),
\end{equation}
when the dimension $k \geq C \cdot \log \alpha(G)$ for some large enough constant $C$,
and this would imply that a scaled version of $h$ will satisfy the constraint with objective value at most $O\big(\vec{\lambda}_2^{e*}(G)\big)$.

The main job is to bound $\nu_h^{e*}(G)$, 
for which we use the arguments in~\cite{JPV22}.
Given $h : V \to \R^k$,
let ${\mathcal B} = \{ uv \in E \mid \norm{h(u)-h(v)}^2 \geq e^{2\eps} \cdot \norm{f(u)-f(v)}^2 \}$
be the set of ``bad edges'' where the projected length is considerably longer than the original length.
We can bound $\nu_h^{e*}(G)$ in terms of the edges in ${\cal B}$ as follows.
For any Eulerian subgraph $A$ that satisfies the constraints in \autoref{prop:lambda2-edge}, twice its objective value is
\begin{eqnarray*}
& & \sum_{uv \notin {\cal B}} A(u,v) \norm{h(u) - h(v)}^2 
+ \sum_{uv \in {\cal B}} A(u,v) \norm{h(u)-h(v)}^2
\\
& = & \sum_{uv \notin {\cal B}} A(u,v) \norm{h(u) - h(v)}^2 
+ \sum_{uv \in {\cal B}} A(u,v) \big( \norm{h(u)-h(v)}^2 - e^{2\eps} \norm{f(u)-f(v)}^2 + e^{2\eps} \norm{f(u)-f(v)}^2 \big)
\\
& \leq & e^{2\eps} \sum_{uv \in E} A(u,v) \norm{f(u)-f(v)}^2
+ \sum_{uv \in {\cal B}} A(u,v) \big( \norm{h(u)-h(v)}^2 - e^{2\eps} \norm{f(u)-f(v)}^2 \big)
\\
& \leq & 2 e^{2\eps} \nu_{f}^{e*}(G) + \sum_{uv \in {\cal B}} w(uv) \big( \norm{h(u)-h(v)}^2 - e^{2\eps} \norm{f(u)-f(v)}^2 \big),
\end{eqnarray*}
where the last inequality is because $A$ is a feasible solution to the SDP in \autoref{prop:lambda2-edge}.
Since the upper bound on the last line no longer depends on $A$, it follows that
\begin{eqnarray*}
\E_h[ 2\nu_h^{e*}(G) ]
& \leq & 2e^{2\eps} \nu_f^{e*}(G) 
+ \E_h \Big[ \sum_{uv \in {\cal B}} w(uv) \big( \norm{h(u)-h(v)}^2 - e^{2\eps} \norm{f(u)-f(v)}^2 \big) \Big]
\\
& = & 2e^{2\eps} \nu_f^{e*}(G) 
+ \sum_{uv \in E} w(uv) \cdot \E_h \big[ \one_{{\cal E}_{u,v}}\big( \norm{h(u)-h(v)}^2 - e^{2\eps} \norm{f(u)-f(v)}^2 \big) \big]
\\
& \leq & 2e^{2\eps} \nu_f^{e*}(G) + e^{-c \eps^2 k} \sum_{uv \in E} w(uv) \norm{f(u)-f(v)}^2
\\
& \leq & 2e^{2\eps} \nu_f^{e*}(G) + 2 e^{-c \eps^2 k} \cdot \alpha(G) \cdot \nu_f^{e*}(G),
\end{eqnarray*}
where the second last inequality is by the second property in \autoref{lem:Gaussian}, and the last inequality is by the large optimal property in \autoref{lem:large-optimal}.
By choosing some constant $\eps \leq 1/4$ and $k \gtrsim \frac{1}{c\eps^2} \log \alpha(G)$, it follows that
\[
\E_h[\nu_h^{e*}(G)]
\leq \big(e^{2\eps} + e^{-c \eps^2 k} \alpha(G) \big) \cdot \nu_f^{e*}(G)
\leq 2 \cdot \nu_f^{e*}(G).
\]
Finally, we lower bound the denominator.
Let ${\cal E}'_v$ be the event that $\norm{h(v)}^2 < e^{-2\eps} \norm{f(v)}^2$.
Using a similar argument as above,
\[
\sum_{v \in V} d(v) \cdot \norm{h(v)}^2 
\geq e^{-2\eps} \sum_{v \in V} d(v) \norm{f(v)}^2 
- \sum_{v \in V} d(v) \cdot \one_{{\cal E}_v} \big(e^{-2\eps} \norm{f(v)}^2 - \norm{h(v)}^2 \big).
\]
We can view the event ${\cal E}'_v$ as ${\cal E}'_{v,0}$ where the zero vector is one of the embedding vectors, so that $\norm{h(v)}^2 < e^{-2\eps} \norm{f(v)}^2$ is equivalent to $\norm{h(v)-h(0)}^2 < e^{-2\eps} \norm{f(v)-f(0)}^2$  .
Thus, we can apply the first property in \autoref{lem:Gaussian} to bound $\E_h[\one_{{\cal E}'_v}] = \Pr[\one_{{\cal E}'_v}]$, so that
\begin{eqnarray*}
\E_h \Big[\sum_{v \in V} d(v) \cdot \one_{{\cal E}_v} \big(e^{-2\eps} \norm{f(v)}^2 - \norm{h(v)}^2 \big)\Big] 
\leq 
\sum_{v \in V} e^{-2\eps} d(v) \norm{f(v)}^2 \cdot \E_h[\one_{{\cal E}'_v}]
 \leq e^{-c \eps^2 k -2\eps} \sum_{v \in V} d(v) \norm{f(v)}^2.
\end{eqnarray*}
By Markov's inequality and the same choice of $\eps$ and $k$, with probability at least $9/10$,
\[
\sum_{v \in V} d(v) \cdot \norm{h(v)}^2
\geq e^{-2\eps} \big(1 - 10 e^{-c \eps^2 k} \big) \sum_{v \in V} d(v) \norm{f(v)}^2
\geq \frac12 \sum_{v \in V} d(v) \norm{f(v)}^2. 
\]
Therefore, \eqref{e:plan} follows by combining the upper bound on the numerator and this lower bound on the denominator.

The proof of the first inequality is the same, with $d(v)$ replaced by $\pi(v)$, $w(uv)$ replaced by $w_{\pi}(uv)$, and with $\alpha(G)$ in the large optimal property replaced by $\Delta \cdot \alpha(G)$ as stated in \autoref{lem:large-optimal}.
\end{proof}

By choosing the best coordinate from a $k$-dimensional embedding, one can achieve the following bound.
The proof is standard and is omitted; see \cite[Proposition 2.9]{OZ22}.

\begin{lemma}[One Dimensional Solution from $k$-Dimensional Solution] \label{lem:k-to-1}
Let $\vec{\lambda}_v^{(k)}(G)$ and $\vec{\lambda}_e^{(k)}(G)$ be as defined in \autoref{def:k-dim-primal}.
Then 
\[
\vec{\lambda}_v^{(1)}(G) \leq k \cdot \vec{\lambda}_v^{(k)}(G)
\quad {\rm and} \quad
\vec{\lambda}_e^{(1)}(G) \leq k \cdot \vec{\lambda}_e^{(k)}(G)
\]
\end{lemma}

\autoref{thm:dimension-reduction} follows immediately from \autoref{thm:JPV} and \autoref{lem:k-to-1}.

\subsection{Rounding Algorithms}  \label{sec:rounding-algorithms}

The main goal in this subsection is to show how to find a set of small directed vertex expansion (respectively directed edge conductance) from a solution to $\vec{\lambda}_v^{(1)}(G)$ (respectively $\vec{\lambda}_e^{(1)}(G)$).

\begin{theorem}[Rounding One Dimensional Solution] \label{thm:rounding}
For any vertex-weighted directed graph $G=(V,E,\pi)$, 
$$\vec{\psi}(G) \lesssim \sqrt{\vec{\lambda}_v^{(1)}(G)}.$$ 
For any edge-weighted directed graph $G=(V,E,w)$,
$$\vec{\phi}(S) \lesssim \sqrt{\vec{\lambda}_e^{(1)}(G)}.$$ 
\end{theorem}

Assuming \autoref{thm:rounding}, we can complete the proofs of the two main results.

\begin{proofof}{\autoref{thm:directed-vertex-expansion} and \autoref{thm:directed-edge-conductance}}
The easy directions are proved in \autoref{prop:directed-vertex-easy} and \autoref{prop:directed-edge-easy}.
For the hard directions, first we solve the semidefinite programs for $\vec{\lambda}_2^{v*}(G)$ in \autoref{prop:lambda2-vertex} and $\vec{\lambda}_2^{e*}(G)$ in \autoref{prop:lambda2-edge}.
Then, we use the dimension reduction result in \autoref{thm:dimension-reduction} to obtain $1$-dimensional solutions to the semidefinite programs with 
$\vec{\lambda}_v^{(1)}(G) \lesssim \log (\Delta \cdot \alpha(G)) \cdot \vec{\lambda}_2^{v*}(G)$ and 
$\vec{\lambda}_e^{(1)}(G) \lesssim \log \alpha(G) \cdot \vec{\lambda}_2^{e*}(G)$.
Then, we apply the rounding result in \autoref{thm:rounding} to establish that
\begin{equation} \label{e:refined2}
\vec{\psi}(G) 
\lesssim \sqrt{\log (\Delta \cdot \alpha(G)) \cdot \vec{\lambda}_2^{v*}(G)}
\quad {\rm and} \quad
\vec{\phi}(G) 
\lesssim \sqrt{\log \alpha(G) \cdot \vec{\lambda}_2^{e*}(G)}.
\end{equation}
Finally, we use the inequality $\alpha(G) \lesssim \Delta / \vec{\psi}(G)$ in \autoref{lem:asymmetric-ratio-vertex} and $\alpha(G) \leq 1 / \vec{\phi}(G)$ in \autoref{lem:asymmetric-ratio-edge} to obtain the final forms in \autoref{thm:directed-vertex-expansion} and \autoref{thm:directed-edge-conductance}.
\end{proofof}

We remark that all the steps in the proofs of the two main results can be implemented in polynomial time, and so these give efficient ``spectral'' algorithms to find a set of small directed vertex expansion or small directed edge conductance.

\subsubsection{Proof Structure and Auxiliary Programs}

The programs $\vec{\lambda}_v^{(1)}$ and $\vec{\lambda}_e^{(1)}$ can be considered ``$\ell_2^2$ programs'' because the embedded distance across an edge is the squared $\ell_2$ distance $\norm{f(u) - f(v)}^2$.
To prove \autoref{thm:rounding}, we first obtain a solution to the following $\ell_1$ versions of $\vec{\lambda}_v^{(1)}$ and $\vec{\lambda}_e^{(1)}$.

\begin{definition}[$\ell_1$ Version of $\vec{\lambda}_v^{(1)}$]
\label{def:l1-vertex}
Given a vertex-weighted directed graph $G = (V, E, \pi)$, let 
  \begin{align*}
    \eta_v(G) :=
     \min_{f: V \rightarrow \R} \max_{A \geq 0} &~~~ \frac12 \sum_{uv \in E} A(u,v) \cdot |f(u) - f(v)|
    \\
    \st&~~~
    A(u, v) = 0 & & \forall uv \not \in E
    \\
    &~~~
    \sum_{v \in V} A(u, v) = \sum_{v \in V} A(v, u)  & & \forall u \in V
    \\
    &~~~
    \sum_{v \in V} A(v, u) = \pi(u) & & \forall u \in V
    \\
    &~~~
    \sum_{v \in V} \pi(v) \cdot f(v) = 0
    \\
    &~~~ \sum_{v \in V} \pi(v) \cdot |f(v)| = 1.
  \end{align*}
\end{definition}

\begin{definition}[$\ell_1$ Version of $\vec{\lambda}_e^{(1)}$]
\label{def:l1-edge}
Given an edge-weighted directed graph $G = (V, E, w)$, let
  \begin{align*}
    \eta_e(G) :=
     \min_{f: V \rightarrow \R} \max_{A \geq 0} &~~~ \frac12 \sum_{uv \in E} A(u, v) \cdot |f(u) - f(v)|
    \\
    \st&~~~
    A(u, v) = 0 & & \forall uv \not \in E
    \\
    &~~~
    \sum_{v \in V} A(u, v) = \sum_{v \in V} A(v, u)  & & \forall u \in V
    \\
    &~~~
    A(u, u) \leq w(uv)  & & \forall uv \in E
    \\
    &~~~
    \sum_{v \in V} d_w(v) \cdot f(v) = 0
    \\
    &~~~ \sum_{v \in V} d_w(v) \cdot |f(v)| = 1.
  \end{align*}
\end{definition}

We will prove in \autoref{sec:l22-to-l1} that there is a square root loss by going from $\ell_2^2$ to $\ell_1$.

\begin{proposition}[Reductions from $\ell_2^2$ to $\ell_1$] \label{prop:l22-to-l1}
For any vertex-weighted directed graph $G = (V, E, \pi)$, 
\[
\eta_v(G) \lesssim \sqrt{ \vec{\lambda}_v^{(1)}(G) }.
\]
For any edge-weighted directed graph $G = (V, E, w)$, 
\[
\eta_e(G) \lesssim \sqrt{ \vec{\lambda}_e^{(1)}(G) }.
\]
\end{proposition}

For threshold rounding, we construct the duals of $\eta_v(G)$ and $\eta_e(G)$ using linear programming duality in the inner maximization problems.

\begin{lemma}[Dual Program of $\eta_v(G)$]
\label{lem:l1-dual-vertex}
Given a vertex-weighted directed graph $G = (V, E, \pi)$, let
  \begin{align*}
    \xi_v(G) :=
     \min_{f: V \rightarrow \R}~\min_{\substack{q: V \rightarrow \R_{\ge 0} \\ r: V \rightarrow \R}} &~~~ \sum_{v \in V} \pi(v) \cdot q(v)
    \\
    \st&~~~
    q(v) \geq |f(u)-f(v)| - r(u) + r(v) & & \forall uv \in E
    \\
    &~~~
    \sum_{v \in V} \pi(v) \cdot f(v) = 0
    \\
    &~~~ \sum_{v \in V} \pi(v) \cdot |f(v)| = 1.
  \end{align*}
Then $\xi_v(G) = 2\eta_v(G)$.
\end{lemma}
\begin{proof}
To write the dual program, we consider the equivalent program of $\vec{\lambda}_v^{(1)}(G)$, where we remove the self-loops and replace the constraint $\sum_{v \in V} A(v,u) = \pi(u)$ by $\sum_{v \in V} A(v,u) \leq \pi(u)$.
Then we multiply the objective of $\eta_v(G)$ by a factor of $2$ (to avoid the factor $1/2$ carrying around).
Then we associate a dual variable $q(u) \geq 0$ to each constraint $\sum_{v \in V} A(v,u) \leq \pi(u)$, and a dual variable $r(u)$ to each constraint $\sum_{v \in V} A(u, v) = \sum_{v \in V} A(v, u)$. The result follows from standard linear programming duality.
\end{proof}

The dual program of $\eta_e(G)$ is constructed in the same way and the proof is omitted.

\begin{lemma}[Dual Program of $\eta_e(G)$]
\label{lem:l1-dual-edge}
Given an edge-weighted directed graph $G = (V, E, w)$, let
  \begin{align*}
    \xi_e(G) :=
     \min_{f: V \rightarrow \R}~\min_{\substack{q: E \rightarrow \R_{\ge 0} \\ r: V \rightarrow \R}} &~~~ \sum_{uv \in E} w(uv) \cdot q(uv)
    \\
    \st&~~~
    q(uv) \geq |f(u)-f(v)| - r(u) + r(v) & & \forall uv \in E
    \\
    &~~~
    \sum_{v \in V} d_w(v) \cdot f(v) = 0
    \\
    &~~~ \sum_{v \in V} d_w(v) \cdot |f(v)| = 1.
  \end{align*}
Then $\xi_e(G) = 2\eta_e(G)$.
\end{lemma}

In \autoref{sec:threshold-rounding}, we will present a threshold rounding algorithm to return a set of small directed vertex expansion (respectively directed edge conductance) from a solution to $\xi_v(G)$ (respectively $\xi_e(G)$), with only a constant factor loss.

\begin{proposition}[Threshold Rounding] \label{prop:threshold-rounding}
For any vertex-weighted directed graph $G = (V, E, \pi)$, 
\[
\vec{\psi}(G) \lesssim \xi_v(G).
\]
For any edge-weighted directed graph $G = (V, E, w)$, 
\[
\vec{\phi}(G) \lesssim \xi_e(G).
\]
\end{proposition}

Note that \autoref{thm:rounding} follows immediately from \autoref{prop:l22-to-l1} and \autoref{prop:threshold-rounding},
so it remains to prove the two propositions in \autoref{sec:l22-to-l1} and \autoref{sec:threshold-rounding}.

\subsubsection[Reduction from l\_2\^2 to l\_1]{Reduction from $\ell_2^2$ to $\ell_1$} \label{sec:l22-to-l1}

We first prove the first inequality in \autoref{prop:l22-to-l1} about directed vertex expansion.
Let $G=(V,E,\pi)$ be a vertex-weighed directed graph.
Let $f: V \to \R$ be a solution to $\vec{\lambda}_v^{(1)}(G)$ with objective value $\lambda_f$, with $A$ being an optimal solution to the inner maximization problem (which can be computed by linear programming).
Our goal is to construct a solution to $\eta_v(G)$ in \autoref{def:l1-vertex} with objective value $O\big(\sqrt{\lambda_f}\big)$.

To this end, define $g: V \to \R$ by
\[
  g(u) := \begin{cases}
    (f(u) + c)^2 & \text{ if } f(u) + c > 0 \\
    -(f(u) + c)^2 & \text{ otherwise },
  \end{cases}
\]
where $c \in \R$ is chosen so as to satisfy the constraint $\sum_u \pi(u) \cdot g(u) = 0$ in \autoref{def:l1-vertex}. 
Note that such $c$ exists and is unique.

We would like to prove that $1 \leq \sum_u \pi(u) \cdot |g(u)| \leq 2$, so that scaling $g$ down by a factor of at most $2$ will satisfy the constraint $\sum_u \pi(u) \cdot |g(u)|=1$ in \autoref{def:l1-vertex}. 
Using $\sum_u \pi(u) f(u) = 0$, it follows that 
  \[
    \sum_u \pi(u) \cdot |g(u)|
    = \sum_u \pi(u) \cdot (f(u) + c)^2
    = \sum_u \pi(u) \cdot f(u)^2 + \pi(V) \cdot c^2
    \ge 1.
  \]
To show that $\sum_u \pi(u) \cdot |g(u)| \le 2$, let
  \[
    s^+(x) := \sum_{u:f(u) + x > 0} \pi(u) \cdot |g(u)|
    \quad {\rm and} \quad
    s^-(x) := \sum_{u:f(u) + x < 0} \pi(u) \cdot |g(u)|.
  \]
Then, note that $s^+(0), s^-(0) \in [0, 1]$, $s^+(x)$ increases as $x$ increases, and $s^-(x)$ decreases as $x$ increases.
If $s^+(0) > s^-(0)$, starting from $x=0$, we decrease $x$ until $x=c$ so that $s^+(c) = s^-(c)$, whence
  \[
    \sum_u \pi(u) \cdot |g(u)| = 2 s^+(c) \le 2 s^+(0) \le 2.
  \]
The case where $s^+(0) \le s^-(0)$ is similar. Therefore, $\sum_u \pi(u) \cdot |g(u)| \in [1,2]$.

Now we bound the objective value of the $\ell_1$ program in \autoref{def:l1-vertex} using $g$ as a solution.
Let $B$ be an optimal solution to the inner maximization problem in \autoref{def:l1-vertex} after fixing $g$.
Assuming the inequality $|g(u)-g(v)|^2 \leq 2(f(u)-f(v))^2 \big( |g(u)| + |g(v)| \big)$ that we will prove below,
the objective value to the $\ell_1$ program is
  \begin{eqnarray*}
    &  &
    \frac12 \sum_{uv \in E} B(u,v) \cdot |g(u) - g(v)|
    \\
    & \lesssim &
    \sum_{uv \in E} B(u,v) \sqrt{(f(u)-f(v))^2 \big(|g(u)| + |g(v)|\big)}
    \\
    & \leq &
    \sqrt{ \sum_{uv \in E} B(u,v) (f(u)-f(v))^2} \cdot \sqrt{\sum_{uv \in E} B(u,v) \big(|g(u)| + |g(v)|\big)}
    \\
    & = &
    \sqrt{ \sum_{uv \in E} B(u,v) (f(u)-f(v))^2 } \cdot \sqrt{ \sum_{u \in V} |g(u)| \cdot \Big( \sum_{v : uv \in E} B(u,v) + \sum_{v : vu \in E} B(v,u) \Big)}
    \\
    & = &
    \sqrt{ \sum_{uv \in E} B(u,v) (f(u)-f(v))^2 } \cdot \sqrt{ 2\sum_{u \in V} \pi(u) \cdot |g(u)| }
    \\
    & \lesssim &
    \sqrt{ \sum_{uv \in E} B(u,v) (f(u)-f(v))^2 } 
    \\
    & \leq &
    \sqrt{ \sum_{uv \in E} A(u,v) (f(u)-f(v))^2 } 
    \\
    & \lesssim &
    \sqrt{\lambda_f},
  \end{eqnarray*}
where the second inequality is by Cauchy-Schwarz, 
the second equality is by the degree constraints in~\autoref{def:l1-vertex},
and the second last inequality is because $A$ is an optimal solution to the inner maximization problem when $f$ is fixed.
Therefore, we conclude that $g$ (after normalizing to satisfy $\sum_{u \in V} \pi(u) \cdot |g(u)| = 1$) is a solution to $\nu_v(G)$ with objective value $O\big(\sqrt{\lambda_f} \big)$.
  
It remains to verify the inequality 
$|g(u)-g(v)|^2 \leq 2(f(u)-f(v))^2 \big( |g(u)| + |g(v)| \big)$.    
There are two cases to consider.
  \begin{itemize}
      \item Case 1: $f(u) + c$ and $f(v) + c$ are of the same sign. 
In this case,
      \begin{eqnarray*}
        |g(u) - g(v)|^2
        &=&
        ((f(u) + c)^2 - (f(v) + c)^2)^2
        \\
        &=&
        (f(u) - f(v))^2 \cdot
        ((f(u) + c) + (f(v) + c))^2
        \\
        &\le&
        2 (f(u) - f(v))^2 \cdot (|g(u)| + |g(v)|).
      \end{eqnarray*}
      \item Case 2: $f(u) + c$ and $f(v) + c$ are of different signs. 
In this case, 
      \begin{eqnarray*}
        |g(u) - g(v)|^2
        &=&
        ((f(u) + c)^2 + (f(v) + c)^2)^2
        \\
        &=&
        ((f(u) + c)^2 + (f(v) + c)^2) \cdot (|g(u)| + |g(v)|)
        \\
        &\le&
        ((f(u) + c) - (f(v) + c))^2 \cdot (|g(u)| + |g(v)|)
        \\
        &=&
        (f(u) - f(v))^2 \cdot (|g(u)| + |g(v)|).
      \end{eqnarray*}
  \end{itemize}
This completes the proof of the first inequality about directed vertex expansion in \autoref{prop:l22-to-l1}.

The proof of the second inequality about directed edge conductance is the same (with $\pi(u)$ replaced by $d(u)$) and is omitted.

\subsubsection{Threshold Rounding} \label{sec:threshold-rounding}

Finally, we prove \autoref{prop:threshold-rounding}.
Again, we first prove the first inequality in \autoref{prop:threshold-rounding} about directed vertex expansion.
Let $G = (V, E, \pi)$ be a vertex-weighted directed graph.
Let $(f,q,r)$ be a feasible solution to $\xi_v(G)$ in \autoref{lem:l1-dual-vertex} with objective value $\xi_{f}$.
Our goal is to construct a nonempty set $S \subset V$ with $\psi(S) \lesssim \xi_{f}$.

The algorithm is a threshold rounding algorithm, where each vertex $u$ is mapped to some $g(u) \in [0,\infty)$ and the output is a set $S_t := \{u \in V \mid g(u) > t\}$ for some threshold $t$.
In previous threshold rounding algorithms for Cheeger-type inequalities, only the embedding function $f: V \to \R$ is used as the function $g$ to produce the output set, so in particular only one ordering of the vertices is considered.

The new twist in our algorithm is that we would consider a few candidate choices for $g(u)$. 
They will all ensure that the threshold rounding would produce a set with small expected directed vertex boundary, and we will choose the one that gives large expected set size.

To this end, define the following four functions:
  \begin{itemize}
      \item $g_1(u) := \max\{ 0, f(u) + r(u) - c_1\}$
      \item $g_2(u) := \max\{0, f(u) - r(u) - c_2\}$
      \item $g_3(u) := \max\{0, -f(u) + r(u) + c_2\}$
      \item $g_4(u) := \max\{0, -f(u) - r(u) + c_1\}$,
  \end{itemize}
  where $c_1$ is a $\pi$-weighted median of $f(u) + r(u)$, so that $\max(\pi(\supp(g_1)), \pi(\supp(g_4))) \le \pi(V)/2$.
  Similarly, $c_2$ is a $\pi$-weighted median of $f(u) - r(u)$, so that $\max(\pi(\supp(g_2)), \pi(\supp(g_3))) \le \pi(V)/2$.\\ 
  
\textbf{Numerator:}
We bound the size of the outer boundary of either $S_t$ or $\overline{S_t}$ for uniformly random $t$, depending on whether the coefficient of $r(u)$ is $-1$ or $+1$ in the function $g_i$.
  
On the one hand, if we consider $g_1$ (similar for $g_3$), then we would bound the expected outer boundary size of $\overline{S_t}$ as:
  \begin{eqnarray*}
    \int_0^{\infty} \pi(\partial^+(\overline{S_t})) \,dt
    &=&
    \sum_v \pi(v)\int_0^{\infty} \one[v \in \partial^+(\overline{S_t})] \, dt
    \\
    &=&
    \sum_v \pi(v) \int_0^{\infty} \one[\exists\, u \text{ with } uv \in E \text{ and } g_1(u) \le t < g_1(v)] \, dt
    \\
    &=&
    \sum_v \pi(v) \max_{u: uv \in E} \{g_1(v) - g_1(u)\}
    \\
    &\le&
    \sum_v \pi(v) \max_{u: uv \in E} \{(f(v) + r(v)) - (f(u) + r(u))\}
    \\
    &\le&
    \sum_v \pi(v) \max_{u: uv \in E} \{|f(u) - f(v)| + r(v) - r(u)\}
    \\
    &\le&
    \sum_v \pi(v) \cdot q(v).
  \end{eqnarray*}
On the other hand, if we consider the function $g_2$ (similar for $g_4$), then we bound the expected outer boundary size of $S_t$ as
  \begin{eqnarray*}
    \int_0^{\infty} \pi(\partial^+(S_t)) \,dt
    &=&
    \sum_v \pi(v) \int_0^{\infty} \one[v \in \partial^+(S_t)] \, dt
    \\
    &=&
    \sum_v \pi(v) \int_0^{\infty} \one[\exists\, u \text{ with } uv \in E \text{ and } g_2(v) \le t < g_2(u)] \, dt
    \\
    &=&
    \sum_v \pi(v) \max_{u: uv \in E} \{g_2(u) - g_2(v)\}
    \\
    &\le&
    \sum_v \pi(v) \max_{u: uv \in E} \{(f(u) - r(u)) - (f(v) - r(v))\}
    \\
    &\le&
    \sum_v \pi(v) \max_{u: uv \in E} \{|f(v) - f(u)| + r(v) - r(u)\}
    \\
    &\le&
    \sum_v \pi(v) \cdot q(v).
  \end{eqnarray*}
To summarize, when we do threshold rounding with respect to any of $g_1,g_2,g_3,g_4$, it holds that
\[
  \int_0^{\infty} \min \big\{ \pi(\partial^+(S_t)), \pi(\partial^+(\overline{S_t})) \big\} \,dt \le \sum_v \pi(v) q(v).
\]

\textbf{Denominator:}
For the function $g_i$, the expected size of $S_t$ is given by
\[
  \int_0^{\infty} \pi(S_t) \,dt
  = \sum_u \pi(u) \int_0^{\infty} \one[g_i(u) > t] \,dt
  = \sum_u \pi(u) \cdot g_i(u).
\]
Therefore, our goal is to show that there exists $1 \leq i \leq 4$ with $\sum_u \pi(u) g_i(u) \ge \Omega(1)$. 
To do so, we will show that
\[
  \sum_{i=1}^4 \sum_u \pi(u) \cdot g_i(u) \ge \Omega(1).
\]
Note that, for any $u \in V$,
\[
  g_1(u) + g_4(u) = \max \{ 0, f(u) + r(u) - c_1 \} + \max\{ 0, -f(u) - r(u) + c_1 \} = |(f(u) + r(u)) - c_1|,
\]
and
\[
  g_2(u) + g_3(u) = \max \{ 0, f(u) - r(u) - c_2 \} + \max\{0, -f(u) + r(u) + c_2\} = |(f(u) - r(u)) - c_2|.
\]
Thus it suffices to show that
\[
  \sum_u \pi(u) \cdot \Big( \big|(f(u) + r(u)) - c_1\big| + \big|(f(u) - r(u)) - c_2\big| \Big) \ge \frac{1}{2}.
\]
To this end, we note that either $\sum_u \pi(u) |f(u) + r(u)| \ge 1$ or $\sum_u \pi(u) |f(u) - r(u)| \ge 1$, because
\[
  \sum_u \pi(u) \left( |f(u) + r(u)| + |f(u) - r(u)| \right)
=  \sum_u \pi(u) \cdot 2 \max(|f(u)|, |r(u)|)
\geq  2 \sum_u \pi(u) |f(u)| = 2.
\]
Assume without loss that $\sum_u \pi(u) \cdot r(u) = 0$ (as we can shift every $r(u)$ by the same amount without changing anything). 
Then both $\sum_u \pi(u) (f(u) + r(u)) = 0$ and $\sum_u \pi(u) (f(u) - r(u)) = 0$. 

Consider first the case where $\sum_u \pi(u) |f(u) + r(u)| \ge 1$; the other case is treated similarly.
Then, since $\sum_u \pi(u)(f(u) + r(u)) = 0$ and $\sum_u \pi(u)|f(u) + r(u)| \ge 1$, it follows that
\[
  \sum_{u: f(u) + r(u) \le 0} \pi(u) |f(u) + r(u)|
  = \sum_{u: f(u) + r(u) \ge 0} \pi(u) |f(u) + r(u)|
  = \frac{1}{2} \sum_u \pi(u) |f(u) + r(u)|
  \ge \frac{1}{2}.
\]
If $c_1 \ge 0$, then
\[
  \sum_u \pi(u) |(f(u) + r(u)) - c_1|
  \ge
  \sum_{u: f(u) + r(u) \le 0} \pi(u) |(f(u) + r(u)) - c_1|
  \ge
  \sum_{u: f(u) + r(u) \le 0} \pi(u) |f(u) + r(u)|
  \ge
  \frac{1}{2},
\]
and similarly if $c_1 < 0$, then
\[
  \sum_u \pi(u) |(f(u) + r(u)) - c_1|
  \ge
  \sum_{u: f(u) + r(u) \ge 0} \pi(u) |(f(u) + r(u)) - c_1|
  \ge
  \sum_{u: f(u) + r(u) \ge 0} \pi(u) |f(u) + r(u)|
  \ge
  \frac{1}{2}.
\]
To summarize,
\[
  \sum_{i=1}^4 \sum_u \pi(u) \cdot g_i(u) =
  \sum_u \pi(u) \cdot \Big( \big|(f(u) + r(u)) - c_1\big| + \big|(f(u) - r(u)) - c_2\big| \Big) \ge \frac{1}{2}.
\]

\textbf{Conclusion:}
There exists $g = g_i$ for some $1 \leq i \leq 4$, such that if we use this function for threshold rounding,
\begin{itemize}
  \item $\int_0^{\infty} \min \big\{ \pi(\partial^+(S_t)), \pi(\partial^+(\overline{S_t})) \big\} \,dt \le \sum_v \pi(v) \cdot q(v) = \xi_{f}$;
  \item $\int_0^{\infty} \pi(S_t) \, dt \ge 1/8$;
  \item $\pi(S_t) \le \pi(V)/2$ always.
\end{itemize}
Hence, we can return some $S = S_t$, whence $0 < \pi(S) \le \pi(V) / 2$ and
\[
  \vec{\psi}(S) 
  =
  \frac{\min \big\{ \pi(\partial^+(S)), \pi(\partial^+(\overline{S})) \big\}}{\min \big\{\pi(S), \pi(\overline{S}) \big\}}
  =
  \frac{\min(\pi(\partial^+(S)), \pi(\partial^+(\overline{S})))}{\pi(S)}
  \le
  8 \xi_{f}.
\]

The proof of the second inequality about directed edge conductance is the same (with the numerator $\sum_{v} \pi(v) \cdot q(v)$ replaced by $\sum_{uv \in E} w(uv) \cdot q(uv)$ and the denominator $\sum_v \pi(v) \cdot |f(v)|$ replaced by $\sum_v d_w(v) \cdot |f(v)| = 1$) and is omitted.

\subsection{Fastest Mixing Time} \label{sec:fastest-mixing}

The goal of this subsection is to prove \autoref{thm:fastest-mixing} that
\[
\frac{1}{\vec{\psi}(G)} \cdot \frac{1}{\log(1/\pi_{\min})}\lesssim \tau^*(G) \lesssim \frac{1}{\vec{\psi}(G)^2} \cdot \log \frac{\Delta}{\vec{\psi}(G)} \cdot \log \frac{1}{\pi_{\min}}.
\]
There are two parts of the proof.
In the first part, we upper bound the fastest mixing time using \autoref{thm:chung-mixing-time} by Fill~\cite{Fil91} and Chung~\cite{Chu05}.
In the second part, we lower bound the fastest mixing time using a combinatorial argument and the $\infty$-norm mixing time that we will define.

\begin{proofof}{\autoref{thm:fastest-mixing}}
Recall that in the setting of the theorem, $\pi$ is not only a weight function, but a probability distribution.
We assume the graph is strongly connected and so $\vec{\lambda}_2^{v*}(G) > 0$.

To prove the upper bound, we prove that $\tau^*(G) \lesssim \big(\vec{\lambda}_2^{v*}(G)\big)^{-1} \cdot \log (\pi_{\min}^{-1})$, and then the result will follow from \autoref{thm:directed-vertex-expansion}.
Let $A$ be an optimal reweighted Eulerian subgraph in \autoref{def:directed-vertex-primal}.
Let $P := \Pi^{-1}A$ be the transition matrix of the ordinary random walk corresponding to the reweighted subgraph $A$.
Observe that $P := \Pi^{-1}A$ is a feasible solution to \autoref{def:fastest-mixing}, and so is $(I+P)/2$. 
Therefore, by \autoref{thm:chung-mixing-time}, 
\[
\tau^*(G) 
\leq \tau\Big( \frac{I+P}{2} \Big)
\lesssim \frac{1}{\lambda_2(\tilde{\L})} \cdot \log \Big( \frac{1}{\pi_{\min}} \Big)
= \frac{1}{\vec{\lambda}_2^{v*}(G)} \cdot \log \Big( \frac{1}{\pi_{\min}} \Big),
\]
where the last inequality is because 
$\tilde{\L} = I - \Pi^{-\frac12} (A+A^T) \Pi^{-\frac12}/2$ as defined in \eqref{e:Chung-Laplacian}
and $\lambda_2(\tilde{\L}) = \vec{\lambda}_2^{v*}(G)$ by \autoref{def:directed-vertex-expansion}.

To prove the lower bound, we consider the $\infty$-norm $\eps$-mixing time defined as 
\[
  \tau^{\infty}_{\eps}(P) := \min \bigg\{t : \max_{p_0: V \rightarrow \R_{\ge 0}} \max_{v \in V} \Big\{ 1 - \frac{p_t(v)}{\pi(v)} \Big\} < \eps \bigg\},
\]
where $p_0$ is an initial distribution on $V$ and $p_t$ denotes $p_0 P^t$.
We will prove that for any feasible solution $P$ to \autoref{def:fastest-mixing},
\begin{equation} \label{e:lower-bound}
\frac{1}{\vec{\psi}(G)} \lesssim \tau^{\infty}_{1/e}(P),
\end{equation}
and this would imply that
\[
\frac{1}{\vec{\psi}(G)} \lesssim \max_P \tau^{\infty}_{1/e}(P) 
\leq \max_P \tau_{1/e}(P) \cdot \log\Big(\frac{1}{\pi_{\min}}\Big) 
= \tau^*(G) \cdot \log\Big(\frac{1}{\pi_{\min}}\Big),
\]
proving the lower bound,  
where the second inequality is by \cite[Proposition~2.47(f)]{Gan06} relating $\tau_{\eps}^{\infty}$ and $\tau_{\eps}$ using sub-multiplicity of mixing time.

To prove \eqref{e:lower-bound}, 
let $P$ be an arbitrary feasible solution to \autoref{def:fastest-mixing},
and $S \subset V$ be a nonempty subset such that $\vec{\psi}(G) = \vec{\psi}(S)$.
We will use $S$ to define an initial distribution $p_0: V \rightarrow \R_{\ge 0}$ such that
\[
    \Delta_{\infty}(p_t, \pi) := \max_{v \in V} \bigg\{ 1 - \frac{p_t(v)}{\pi(v)}\bigg\} > \frac{1}{e}
\]
for any $t \leq 1/(4 \vec{\psi}(S) )$, and this would imply that $\tau_{1/e}^{\infty}(P) > 1/(4 \vec{\psi}(S))$.

To define $p_0$,  
we assume without loss of generality that $\pi(S) \le 1/2$ and consider two cases.

\begin{enumerate}

      \item $\pi(\partial^+(S)) \le \pi(\partial^+(\overline{S}))$. 
      In this case, we set
      \[
        p_0(u) = \begin{cases}
          \pi(u) / \pi(S), \text{ if } u \in S; \\
          0, \text{ otherwise.}
        \end{cases}.
      \]
      We will show that
        $p_t(S) := \sum_{v \in S} p_t(v)
        \ge
        1 - t \cdot \vec{\psi}(S)$
      for all $t \ge 0$.
      Note that, by induction, $p_t(v) \le \pi(v) / \pi(S)$ for all $v \in V$ and $t \ge 0$, as
      \[
        p_{t+1}(v) = \sum_{u \in V} p_{t}(u) \cdot P(u, v) \le \sum_{u \in V} \frac{\pi(u)}{\pi(S)} \cdot P(u, v) = \frac{\pi(v)}{\pi(S)}.
      \]
      It follows that at step $t+1$, the total amount of probability mass escaping from $S$ is at most
      \[
        \sum_{v \in \partial^+(S)} p_t(v) \le \frac{\pi(\partial^+(S))}{\pi(S)} = \vec{\psi}(S).
      \]
      Hence, for any $t \le 1/(4 \vec{\psi}(S))$, we have $p_t(\overline{S}) \le \frac{1}{4} \le \frac{1}{2} \cdot \pi(\overline{S})$, and so
      \[
        \Delta_{\infty}(p_t, \pi) \ge \max_{v \in \overline{S}} \Big\{ 1 - \frac{p_t(v)}{\pi(v)} \Big\} \ge 1 - \frac{p_t(\overline{S})}{\pi(\overline{S})} \ge \frac12 > \frac1e.
      \]

      \item $\pi(\partial^+(S)) > \pi(\partial^+(\overline{S}))$. 
      In this case, we define
      \[
        p_0(u) = \begin{cases}
          \pi(u) / \pi(\overline{S}), \text{ if } u \not\in S; \\
          0, \text{ otherwise.}
        \end{cases}.
      \]
      We will show that $p_t(S) \le 2t \cdot \pi(S) \cdot \vec{\psi}(S)$.
      Again, by induction, $p_t(v) \le \pi(v) / \pi(\overline{S})$ for all $v \in V$ and $t \ge 0$. 
      It follows that in step $t+1$, the total amount of probability mass entering $S$ is at most
      \[
        \sum_{v \in \partial^+(\overline{S})} p_t(v) \le \frac{\pi(\partial^+(\overline{S}))}{\pi(\overline{S})}
        \le 2 \pi(S) \cdot
        \frac{\pi(\partial^+(\overline{S}))}{\pi(S)}
        = 2 \pi(S) \cdot \vec{\psi}(S).
      \]
      Hence, for $t \le 1/(4 \vec{\psi}(S))$, 
      we have $p_t(S) \le \pi(S) / 2$, and so
      \[
        \Delta_{\infty}(p_t, \pi) \ge \max_{v \in S} \Big\{ 1 - \frac{p_t(v)}{\pi(v)} \Big\} \ge 1 - \frac{p_t(S)}{\pi(S)} \ge \frac12 > \frac1e.
      \]
  \end{enumerate}
This completes the proof of the lower bound and hence \autoref{thm:fastest-mixing}.
\end{proofof}

\subsection{Relations with Previous Work} \label{sec:relations}

In this subsection, we show some examples where the Cheeger constant~\cite{Fil91,Chu05} is very different from directed edge conductance and directed vertex expansion, and relate the semidefinite program in~\cite{ACMM05} to the one for reweighted eigenvalue in \autoref{def:directed-edge-primal}.

\subsubsection{Cheeger Constant, Edge Conductance, and Vertex Expansion} \label{sec:Chung}

We show two examples.  
In the first example, the Cheeger constant in \eqref{e:Cheeger-constant} is large while the directed edge conductance and directed vertex expansion is small.

\begin{example}[Large Cheeger Constant but Small Edge Conductance and Vertex Expansion]
Consider the directed graph shown in the figure.
\begin{figure}[H]
    \centering
    \includegraphics[width=0.5\textwidth]{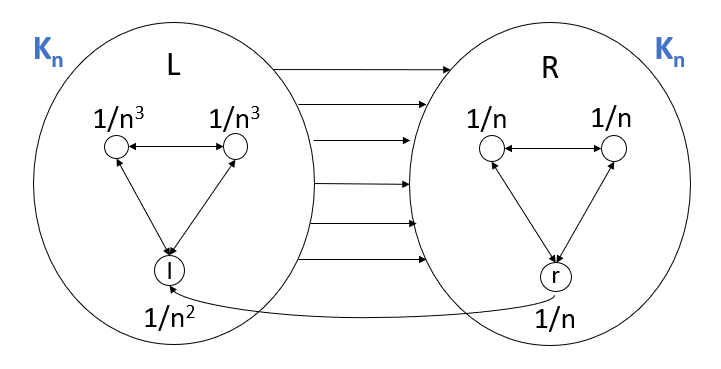}
    \label{fig:2-clique}
\end{figure}
Both $L$ and $R$ are cliques of size $n$.
There is a directed edge from every vertex in $L$ to every vertex in $R$.
There is a directed edge $rl$ from a special vertex $r \in R$ to a special vertex $l \in L$.

When the graph $G$ has the same weight on each edge and the same weight on each vertex,
it is clear that $\vec{\phi}(G) \leq 1/n^2$ and $\vec{\psi}(G) \leq 1/n$ as there is only one directed edge from $R$ to $L$.

We claim that the Cheeger constant $h(G)$ in \eqref{e:Cheeger-constant} is $\Omega(1)$.
The reason is that the Cheeger constant is normalized by the probabilities in the stationary distribution $\pi$, and this will make $L$ to have small $\pi$-weight and so both $h(L)$ and $h(R)$ become big after the normalization. 
More precisely, after some calculations that are omitted, we have $\pi(v) \approx 1/n$ for every vertex $v \in R$, $\pi(u) \approx 1/n^3$ for every vertex $u \in L-\{l\}$, and $\pi(l) \approx 1/n^2$.
This implies that $h(L)=h(R)=\Omega(1)$, and indeed $h(G) = \Omega(1)$ after a case analysis which we omit. 
\end{example}

This example shows that the edge conductance of the reweighted subgraph with respect to the stationary distribution does not provide a good approximation to directed edge conductance and directed vertex expansion, while an optimal reweighted subgraph does identify the bottlenecks in the directed graph.

In the second example, the $k$-way Cheeger constant is large but the $k$-way directed edge conductance is small.

\begin{example}[Large $k$-Way Cheeger Constant but Small $k$-Way Edge Conductance]
Let $G$ be a directed cycle over the vertex set $[n]$. 
For each $i\in \{2,3,...n-2\}$ we add an extra edge $(i,n)$.
\begin{figure}[H]
    \centering
    \includegraphics[width=0.45\textwidth]{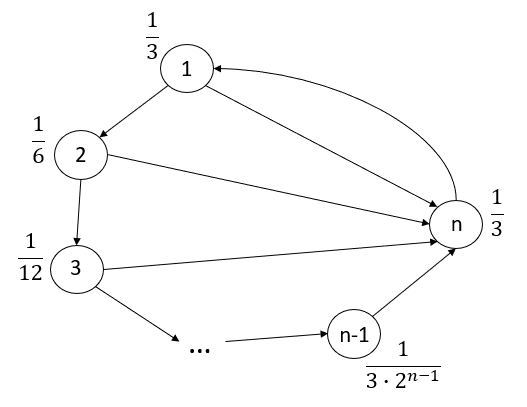}
    \label{fig:fast-dropping-cycle}
\end{figure}
The figure of the graph is shown with the stationary distribution of the ordinary random walk on the graph where every edge has the same weight.
In this example, the $k$-way directed edge conductance is $k/n$, 
but the graph has Cheeger constant $\Omega(1)$ because the vertices $\{2,3,...n-1\}$ have exponentially decreasing stationary weight. 
\end{example}

Since large Cheeger constant implies large $\lambda_2(\tilde{L})$ implies large $\lambda_k(\tilde{L})$ ($\tilde{L}$ is the Laplacian defined by Chung in \eqref{e:Chung-Laplacian}), this example shows that $\lambda_k(\tilde{L})$ is large  but the $k$-way directed edge conductance is small.
This rules out the possibility of having a higher-order Cheeger inequality for directed graphs relating $\lambda_k(\tilde{\L})$ to $k$-way directed edge conductance.

\subsubsection{Semidefinite Program for Directed Sparsest Cut} \label{sec:ACMM}

We compare the semidefinite program for $\vec{\lambda}_2^{e*}(G)$ in \autoref{prop:lambda2-edge} with the semidefinite program for the directed sparsest cut probelm in~\cite{ACMM05}.
Given a directed graph $G=(V,E)$ with edge weights $w: E \rightarrow \R_{\ge 0}$, the directed sparsest cut problem is defined as
\begin{align*}
    \varphi(G) := \min_{S\subseteq V}\frac{\min\big\{w(\delta^+(S)), w(\delta^+(\overline{S}))\big\}}{\min\{|S|,|\overline{S}|\}}.
\end{align*}
Agarwal, Charikar, Makarychev, and Makarychev~\cite{ACMM05} gave a semidefinite program relaxation ${\rm sdp}_{\varphi}$ for $\varphi(G)$ and proved that ${\rm sdp}_{\varphi} \lesssim \varphi(G) \lesssim \sqrt{\log |V| \cdot {\rm sdp}_{\varphi}}$.  

We note that ${\rm sdp}_{\varphi}$ can be modified slightly to give a similar approximation to the directed edge conductance $\vec{\phi}(G)$ in \autoref{def:directed-edge-conductance}.
Consider the semidefinite program
  \begin{align*}
    {\rm sdp}_{\vec{\phi}}  :=
     \min_{f: V\cup \{x\}\to \mathbb{R}^n} &~~~  \sum_{uv \in E} w(u,v)\big(\norm{f(u) - f(v)}^2 - \norm{f(u) - f(x)}^2 + \norm{f(v) - f(x)}^2\big)     \\
    \st&~~~
    \norm{f(u) - f(v)}^2 + \norm{f(v) - f(w)}^2 \geq \norm{f(u) - f(w)}^2  \quad \quad \forall u, v, w \in V \cup \{x\}
    \\
    &~~~
    \sum_{v \in V} d(v) \cdot f(v) = \vec{0}
    \\
    &~~~ \sum_{v \in V} d(v) \cdot \norm{f(v)}^2 = 1,
  \end{align*}
The only difference between ${\rm sdp}_{\vec{\phi}}$ and ${\rm sdp}_{\varphi}$ is the scaling of $f(v)$ by $d(v)$ (instead of $1$ for ${\rm sdp}_{\varphi}$), which corresponds to the degree weights in the denominator of the directed edge conductance in \autoref{def:directed-edge-conductance}.
We note that a simple modification of the proof in \cite{ACMM05} shows that 
${\rm sdp}_{\vec{\phi}} \lesssim \vec{\phi}(G)\lesssim \sqrt{\log |V|}\cdot{\rm sdp}_{\vec{\phi}}$.  

To our knowledge, it was not known that ${\rm sdp}_{\vec{\phi}}$ can be used to certify whether a directed graph has constant edge conductance as in \autoref{thm:directed-edge-conductance}, as the analysis using triangle inequalities based on~\cite{ARV09} has a $\sqrt{\log |V|}$ factor loss. 
However, we observe that the semidefinite program in \autoref{prop:lambda2-edge} for $\vec{\lambda}_2^{e*}(G)$ is a weaker program than ${\rm sdp}_{\vec{\phi}}$.

\begin{claim}[$\vec{\lambda}_2^{e*}(G)$ and ${\rm sdp}_{\vec{\phi}}$] \label{c:ACMM}
For any directed graph $G=(V,E)$ with a weight function $w : E \to \R_{\geq 0}$,
it holds that $\vec{\lambda}_2^{e*}(G) \leq {\rm sdp}_{\vec{\phi}}$.
\end{claim}
\begin{proof}
Consider the following equivalent characterization of $\vec{\lambda}_2^{e*}(G)$ by using LP duality in the inner maximization problem as in \autoref{lem:l1-dual-edge}:
  \begin{align*}
    \vec{\lambda}_2^{e*}(G) =
     \min_{f: V \rightarrow \R^n}~\min_{\substack{q: E \rightarrow \R_{\ge 0} \\ r: V \rightarrow \R}} &~~~ \sum_{uv \in E} w(uv) \cdot q(uv)
    \\
    \st&~~~
    q(uv) \geq \norm{f(u)-f(v)}^2 - r(u) + r(v) & & \forall uv \in E
    \\
    &~~~
    \sum_{v \in V} d(v) \cdot f(v) = 0
    \\
    &~~~ \sum_{v \in V} d(v) \cdot \norm{f(v)}^2 = 1.
  \end{align*}
We will show that for every feasible solution $f:V\cup \{x\}\rightarrow \mathbb{R}^n$ to ${\rm sdp}_{\vec{\phi}}$, there is a feasible solution $f':V\rightarrow \mathbb{R}^n, q: E \rightarrow \R_{\ge 0} ,r:V \rightarrow \R$ to $\vec{\lambda}_2^{e*}(G)$ with the same objective value.
Then the lemma would follow immediately.
To this end, define $f'(v) = f(v)$ for $v\in V$, $r(v) = \norm{f(v) - f(x)}^2$ for $v \in V$ and $q(u,v) = \norm{f(u) - f(v)}^2 -r(u) + r(v)$ for $uv \in E$. 
Clearly, the objective values are equal. 
Also, we see that the constraints $q(u,v) \geq 0$ are satisfied because of the triangle inequalities. 
\end{proof}

Therefore, \autoref{thm:directed-edge-conductance} and \autoref{c:ACMM} imply that
\[
{\rm sdp}_{\vec{\phi}} \lesssim \vec{\phi}(G)
\lesssim \sqrt{ {\rm sdp}_{\vec{\phi}} \cdot \log{\frac{1}{\vec{\phi}(G)}} },
\]

where the ``easy direction'' ${\rm sdp}_{\vec{\phi}} \lesssim \vec{\phi}(G)$ follows because ${\rm sdp}_{\vec{\phi}}$ is a relaxation of directed vertex expansion $\vec{\phi}(G)$.
This provides a new analysis that ${\rm sdp}_{\vec{\phi}}$ can also be used to certify constant edge conductance in directed graphs.

\section{Generalizations of Cheeger Inequalities for Directed Graphs}
\label{sec:supporting-results}

For undirected graphs, there are several interesting generalizations of Cheeger's inequality: Trevisan's result that relates $\lambda_n$ to bipartite edge conductance~\cite{Tre09}, the higher-order Cheeger's inequality that relates $\lambda_k$ to $k$-way edge conductance~\cite{LOT12,LRTV12}, and the improved Cheeger's inequality that relates $\lambda_2$ and $\lambda_k$ to edge conductance~\cite{KLLOT13}.
Using reweighted eigenvalues for vertex expansion, close analogs of these results were obtained in~\cite{KLT22}, relating $\lambda_n^*$ to bipartite vertex expansion, $\lambda_k^*$ to $k$-way vertex expansion, and $\lambda_2^*$ and $\lambda_k^*$ to vertex expansion. 

In this section, we study whether there are close analogs of these results for directed graphs, using reweighted eigenvalues for directed vertex expansion in \autoref{def:directed-vertex-expansion} and directed edge conductance in \autoref{def:directed-edge-conductance}.
Perhaps surprisingly, we show that the natural analogs of Trevisan's result and higher-order Cheeger's inequality do not hold, but we obtain analogs of the improved Cheeger's inequality for directed vertex expansion and directed edge conductance.

\subsection{Higher-Order Cheeger Inequality for Directed Graphs} \label{sec:higher-order-directed}

Given an undirected graph $G=(V,E)$ with a weight function $w : E \to \R_{\geq 0}$,
let $\lambda_k$ be the $k$-th smallest eigenvalue of the normalized Laplacian matrix of $G$.
A basic result in spectral graph theory states that $\lambda_k=0$ if and only if $G$ has at least $k$ connected components, or equivalently $G$ has at least $k$ disjoint subsets $S_1, \ldots, S_k$ each with edge conductance zero.

The higher-order Cheeger inequality is a robust generalization of this basic result.
Define the $k$-way edge conductance of $G$ as
$\phi_k(G) := \min_{S_1, S_2, \ldots, S_k} \max_{1 \leq i \leq k} \phi(S_i)$ where $S_1, S_2, \ldots, S_k$ are over pairwise disjoint subsets of $V$.
Then the basic result is that $\lambda_k = 0$ if and only if $\phi_k(G)=0$,
and the higher-order Cheeger's inequality~\cite{LOT12,LRTV12} is that 
\begin{equation} \label{e:higher-order}
\lambda_k \lesssim \phi_{k}(G) \lesssim k^2 \sqrt{\lambda_k} 
\quad {\rm~and~} \quad
\phi_{k/2}(G) \lesssim \sqrt{\lambda_k \log k}.
\end{equation}
For a directed graph $G=(V,E)$, we can define $\vec{\lambda}_k^{v*}(G)$ and $\vec{\lambda}_k^{e*}(G)$ as in \autoref{def:directed-vertex-expansion} and \autoref{def:directed-edge-conductance}, but with the objective function replaced by maximizing the $k$-th smallest eigenvalue.
The following is an analog of the basic result.

\begin{proposition}[Reweighted Eigenvalues and Strongly Connected Components] \label{prop:lambda-k-SCC} \phantom{ }
\begin{itemize}
\item
For any directed graph $G = (V,E)$ with weight function $w : E \to \R_{\geq 0}$, then $\vec{\lambda}_k^{e*}(G) = 0$ if and only if $G$ has at least $k$ strongly connected components.
\item
For any directed graph $G = (V,E)$ with weight function $\pi : V \to \R_{\geq 0}$, then $\vec{\lambda}_k^{v*}(G) = 0$ if and only if $G$ has at least $k$ strongly connected components.
\end{itemize}
\end{proposition}
\begin{proof}
In one direction, assume $G$ has at least $k$ strongly connected components $S_1, \ldots, S_k$.
Then, in any Eulerian reweighted subgraph $A$, we claim that $\sum_{uv \in \delta^+(S_i)} A(u,v) = \sum_{uv \in \delta^-(S_i)} A(u,v) = 0$ for $1 \leq i \leq k$.
To see this, suppose to the contrary that $uv \in \delta^+(S_i)$ and $A(u,v) > 0$, then as the edge set of any Eulerian graph can be decomposed into edge disjoint cycles, there must be a directed cycle $C$ with $uv \in C$, but then $S_i \cup C \supseteq S_i \cup \{v\}$ is also strongly connected, contradicting that $S_i$ is a maximally strongly connected subset.  
Therefore, in the underlying undirected graph defined by $\frac12 (A+A^T)$, each $S_i$ is a set of conductance zero, and thus $\lambda_k = 0$ by the basic fact.
Since this holds for any Eulerian reweighted subgraph $A$, it follows that $\vec{\lambda}_k^{v*}(G) = \vec{\lambda}_k^{e*}(G) = 0$.

In the other direction, assume $G$ has less than $k$ strongly connected components $S_1, \ldots, S_l$ for $l < k$.
Then, in each strongly connected component $S_i$, there is an Eulerian reweighting $A_i$ in the induced subgraph of $S_i$ such that $S_i$ is strongly connected.
(It is not difficult to see this directly, or one can use Hoffman's result in \autoref{lem:Hoffman}.)
So, there is an Eulerian reweighting such that the underlying undirected graph $G'$ has at most $l < k$ connected components, 
and thus $\vec{\lambda}_k^{v*}(G), \vec{\lambda}_k^{e*}(G) \geq \lambda_k(G') > 0$ by the basic result.
\end{proof}

One might expect that there is a robust generalization of \autoref{prop:lambda-k-SCC} relating $\vec{\lambda}_k^{v*}(G)$ and $\vec{\lambda}_k^{e*}(G)$ to $k$-way directed vertex expansion and $k$-way directed edge conductance, just as in the case $k=2$ in \autoref{thm:directed-vertex-expansion} and \autoref{thm:directed-edge-conductance}.
But in general, unlike undirected graphs, it is not true that $G$ has at least $k$ strongly connected components if and only if $G$ has at least $k$ disjoint subsets $S_1, \ldots, S_k$ each with directed edge conductance zero or directed vertex expansion zero.
Note the subtlety that this is true for $k=2$, as there is a source component and a sink component with directed edge conductance and directed vertex expansion zero.

\begin{example}[Counterexample to Higher-Order Cheeger Inequality for Directed Graphs] \label{ex:higher-order}
Consider the complete directed acyclic graph $G$ where the vertex set is $[n]$ and there is a directed edge $ij$ for every $i < j$.
On the one hand, $\vec{\lambda}_k^{v*}(G) = \vec{\lambda}_k^{e*}(G) = 0$ for every $k \le n$, as any Eulerian reweighting must have $n$ isolated vertices (with self-loops).
On the other hand, for any $k \geq 3$, at least one set has non-zero directed edge conductance.  
Furthermore, it can be shown that for $k \geq 2 \log_2 n$, any $k$ disjoint subsets must contain at least one subset of directed edge conductance at least $1/4$.
This provides a strong counterexample that $ \vec{\lambda}_k^{e*}(G)$ is small but the $k$-way directed edge conductance $\vec{\phi}_k(G)$ is large. 
A similar argument can be made for the case of directed vertex expansion.
\end{example}

We believe that there is still a robust generalization of \autoref{prop:lambda-k-SCC}, such that $\vec{\lambda}_k^{v*}(G), \vec{\lambda}_k^{e*}(G)$ is small if and only if there are $k$ disjoint subsets where each is ``close'' to a strongly connected component.
But it is not clear how to formulate closeness to a strongly connected component, as it is a ``global'' property that cannot be determined by only looking at the edges incident to a subset $S \subseteq V$.
On a technical level, we remark that the proofs in~\cite{LOT12,KLT22} can be followed to construct $k$ disjointly-supported functions $f_1,\ldots,f_k$ from a solution to $\vec{\lambda}_k^{v*}(G)$ and $\vec{\lambda}_k^{e*}(G)$, such that each $f_i$ has small objective value to $\vec{\lambda}_2^{v*}(G)$ and $\vec{\lambda}_2^{e*}(G)$.  
However, using the new threshold rounding algorithm for $\vec{\lambda}_2^{v*}(G)$ and $\vec{\lambda}_2^{e*}(G)$ based on $f_i \pm r_i$ in \autoref{sec:threshold-rounding}, we can no longer conclude that there is a subset $S_i$ of small directed edge conductance or directed vertex expansion in the support of $f_i$, as the support of $f_i \pm r_i$ could be very different from that of $f_i$.
This also indirectly shows that the new idea of doing threshold rounding on $f \pm r$ is a necessary modification.

We leave the problem of proving a robust generalization of \autoref{prop:lambda-k-SCC} as an open problem.

\subsection{Bipartite Cheeger Inequality for Directed Graphs}

Another basic result in spectral graph theory states that $\lambda_n=2$ if and only if $G$ has a bipartite component $S$, or equivalently $G$ has a set $S$ of conductance zero with the induced subgraph $G[S]$ being bipartite.
Trevisan~\cite{Tre09} proved a robust generalization of this basic result, by proving a Cheeger-type inequality that $\lambda_n$ is close to $2$ if and only if $G$ has a set $S$ of small conductance with the induced subgraph $G[S]$ being close to bipartite.

As in \autoref{sec:higher-order-directed}, we can use the $n$-th reweighted eigenvalue to prove an analog of the basic result for directed graphs.
We omit the proof as it is similar to that in \autoref{prop:lambda-k-SCC} and also because it is not used in other results.

\begin{proposition}[Reweighted Eigenvalues and Bipartite Strongly Connected Components] \label{prop:lambda-n-SCC}
\phantom{ }
\begin{itemize}
  \item For any directed graph $G = (V,E)$ with weight function $w : E \to \R_{\geq 0}$, $\vec{\lambda}_n^{e*}(G) = 2$ if and only if $G$ has a strongly connected component $S$ such that the induced subgraph $G[S]$ is bipartite.
  \item For any directed graph $G = (V,E)$ with weight function $\pi : V \to \R_{\geq 0}$, $\vec{\lambda}_n^{v*}(G) = 2$ if and only if $G$ has a strongly connected component $S$ such that the induced subgraph $G[S]$ is bipartite.
\end{itemize}
\end{proposition}

As in \autoref{sec:higher-order-directed}, the natural analogs of Trevisan's result for directed graphs are not true, because the existence of a nearly strongly connected bipartite component does not imply the existence of a set $S$ of small directed edge conductance or directed vertex expansion, with the induced subgraph $G[S]$ being close to bipartite.

\begin{example}[Counterexample to Bipartite Cheeger Inequality for Directed Graphs] \label{ex:bipartite-directed}
Consider the example shown in the figure below. 
\begin{figure}[h]
     \centering
     \includegraphics[width=0.38\textwidth]{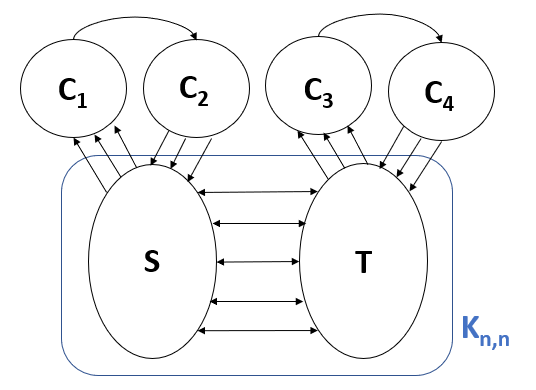}
     \label{fig:bipartite-counter-example}
 \end{figure}
In this directed graph $G$, $|S|=|T|=n$ and $|C_1|=|C_2|=|C_3|=|C_4|=n/2$.
Each $C_i$ is a clique, and there is only one edge from $C_1$ to $C_2$ and only one edge from $C_3$ to $C_4$.
The induced subgraph on $S \cup T$ is a complete bipartite graph $K_{n,n}$.
Every vertex in $S$ has an edge to every vertex in $C_1$, and every vertex in $C_2$ has an edge to every vertex in $S$.
Similarly, every vertex in $T$ has an edge to every vertex in $C_3$, and every vertex in $C_4$ has an edge to every vertex in $T$.
Every edge in $G$ has weight one.

On the one hand, because of the bottlenecks from $C_1$ to $C_2$ and from $C_3$ to $C_4$, any Eulerian reweighing $A$ will have $\sum_{uv : u \in S, v \in C_1 }A(u,v) = \sum_{uv : u \in C_2, v \in S}A(u,v) \leq 1$ and $\sum_{uv: u \in T, v \in C_3 }A(u,v) = \sum_{uv : u \in C_4, v \in T}A(u,v) \leq 1$.
Therefore, $S \cup T$ is an induced bipartite graph with small edge conductance in the underlying undirected graph $\frac12 (A+A^T)$, and one can use the easy direction of Trevisan's result to show that $\vec{\lambda}_n^{e*}(G) \geq 2 - O(1/n^2)$.
On the other hand, $S \cup T$ has large directed edge conductance, and any subset with small directed edge conductance must be far from bipartite because of the edges induced in $C_i$.
We could formally define directed bipartite edge conductance $\vec{\phi}_B(G)$ and show that $\vec{\phi}_B(G) = \Omega(1)$ is large, but we decide to omit these details.
To summarize, this gives a strong counterexample where $\vec{\lambda}_n^{e*}(G)$ is very close to $2$ but there does not exists any subset $S$ with small directed edge conductance and the induced subgraph $G[S]$ being close to bipartite.
A similar argument can be made for the case of directed vertex expansion.
\end{example}

As in \autoref{sec:higher-order-directed}, we believe that there is a robust generalization of \autoref{prop:lambda-n-SCC} that $\vec{\lambda}_n^{e*}(G)$ and $\vec{\lambda}_n^{e*}(G)$ are close to $2$ if and only if $G$ has a nearly bipartite strongly connected component. 
We leave it as an open problem to formulate the combinatorial condition and to prove such a Cheeger-type inequality.

\subsection{Improved Cheeger Inequality for Directed Graphs} \label{sec:improved-Cheeger-directed}

The improved Cheeger inequality in~\cite{KLLOT13} states that $\phi(G) \lesssim k \cdot \lambda_2(G) / \sqrt{\lambda_k(G)}$ for any $k \geq 2$ and any undirected graph $G$.
This shows that $\lambda_2(G)$ is a tighter approximation to $\phi(G)$ when $\lambda_k(G)$ is large for a small $k$.
The result provides an explanation for the good empirical performance of the spectral partitioning algorithm.

The improved Cheeger inequality was extended in~\cite{KLT22} to relate the reweighted second eigenvalue and vertex expansion, showing that
$\psi(G) \lesssim k^{\frac{3}{2}} \cdot \lambda_2^*(G) \cdot \log \Delta / \sqrt{\lambda_k^*(G)}$.

Unlike the higher-order and bipartite Cheeger inequalities, 
we can extend this result to directed graphs as this is only about $2$-way partitioning (recall the discussion above \autoref{ex:higher-order}).
This potentially can also be used to explain the good empirical performance of the spectral algorithm in \autoref{thm:directed-edge-conductance}.

\begin{theorem}[Improved Cheeger's inequality for Directed Vertex Expansion] \label{thm:improved-cheeger-vertex}
Let $G = (V, E, \pi)$ be a vertex-weighted directed graph.
For any $2 \le k \le n$,
  \[
    \vec{\lambda}_2^{v*}(G) \lesssim \vec{\psi}(G) 
    \lesssim \frac{k^{3/2} \cdot \log (\Delta \cdot \alpha(G)) \cdot \vec{\lambda}_2^{v*}(G)}
    {\sqrt{\vec{\lambda}_k^{v*}(G)}}
    \lesssim \frac{k^{3/2} \cdot \log (\Delta / \vec{\psi}(G) ) \cdot \vec{\lambda}_2^{v*}(G)}
    {\sqrt{\vec{\lambda}_k^{v*}(G)}}.
  \]
\end{theorem}

\begin{theorem}[Improved Cheeger's inequality for Directed Edge Conductance] \label{thm:improved-cheeger-edge}
Let $G = (V, E, w)$ be an edge-weighted directed graph.
For any $2 \le k \le n$,
  \[
    \vec{\lambda}_2^{e*}(G) \lesssim \vec{\phi}(G) 
    \lesssim \frac{k^{3/2} \cdot \log \alpha(G) \cdot \vec{\lambda}_2^{e*}(G)}
    {\sqrt{\vec{\lambda}_k^{e*}(G)}}
    \lesssim \frac{k^{3/2} \cdot \log (1 / \vec{\phi}(G) ) \cdot \vec{\lambda}_2^{e*}(G)}
    {\sqrt{\vec{\lambda}_k^{e*}(G)}}.
  \]
\end{theorem}

The proofs of the two results are similar to that in~\cite{KLLOT13,KLT22} and also similar to each other as in \autoref{sec:main-results}, 
so we just provide a sketch of the proof of \autoref{thm:improved-cheeger-vertex} in the following.

Note that $\vec{\lambda}_k^{v*}(G)$ not a convex optimization problem.
As in~\cite{KLT22}, we change the objective in \autoref{def:directed-vertex-primal} to maximize the sum of the $k$ smallest eigenvalues $\sum_{i=1}^k \lambda_i(\L)$, so that we can use \autoref{prop:sum-of-lambda-k} to write this as a semidefinite program, which we call $\vec{\sigma}_k^{v*}(G)$.
Using the same manipulations as in~\cite[Proposition 5.2]{KLT22} and \autoref{prop:lambda2-vertex},
we can write
\begin{align*}
    \vec{\sigma}_k^{v*}(G) :=
     \min_{f: V \rightarrow \R^n} \max_{A \geq 0} &~~~ \frac12 \sum_{uv \in E} A(u,v) \cdot \norm{f(u) - f(v)}^2
    \\
    \st&~~~
    A(u, v) = 0 & & \forall uv \not \in E
    \\
    &~~~
    \sum_{v \in V} A(u, v) = \sum_{v \in V} A(v, u)  & & \forall u \in V
    \\
    &~~~
    \sum_{v \in V} A(v, u) = \pi(u) & & \forall u \in V
    \\
    &~~~
    \sum_{v \in V} \pi(v) \cdot f(v) f(v)^T \preceq I_n
    \\
    &~~~ \sum_{v \in V} \pi(v) \cdot \norm{f(v)}^2 = k.
\end{align*}
The proof will relate $\vec{\lambda}_2^{v*}(G)$ and $\vec{\sigma}_k^{v*}(G)$ to $\vec{\psi}(G)$.
Note that $ \vec{\lambda}_k^{v*}(G)  \leq  \vec{\sigma}_k^{v*}(G)  \leq  k \cdot \vec{\lambda}_k^{v*}(G)$.
This is where an extra factor of $\sqrt{k}$ is lost compared to the bound in~\cite{KLLOT13}.

We follow the same two-step approach in~\cite{KLLOT13,KLT22}.
The first step is to prove that if there is a $1$-dimensional solution to $\vec{\lambda}_v^{(1)}(G)$ that is close to a $k$-step function (i.e. a function with at most $k$ distinct values), then the approximation guarantee of threshold rounding in \autoref{sec:threshold-rounding} is improved.

\begin{proposition}[Improved Threshold Rounding] \label{prop:k-step-rounding-vertex}
Let $G = (V, E, \pi)$ be a vertex-weighted directed graph.
Given a solution $f : V \to \R$ to $\vec{\lambda}_v^{(1)}(G)$ with objective value $\lambda_f$ and a $k$-step function $y_f: V \to \R$ approximating $f$, it holds that
\[
\vec{\psi}(G) \lesssim 
\eta_v(G) \lesssim k \cdot \lambda_f + k \norm{f-y_f}_{\pi} \sqrt{\lambda_f},
\]
where $\eta_v(G)$ is the $\ell_1$-version of $\vec{\lambda}_v^{(1)}(G)$ in \autoref{def:l1-vertex}, and $\norm{z}_{\pi}^2 := \sum_v \pi(v) \cdot z(v)^2$ for any $z : V \to \R$.
Note the first inequality is by \autoref{prop:threshold-rounding}.
\end{proposition}

The second step is to prove that if $\vec{\sigma}_k^{v*}(G)$ is large for a small $k$, then there is a good $k$-step approximation to a good solution to $\vec{\lambda}_v^{(1)}(G)$.
As in \autoref{sec:threshold-rounding}, we consider the $\ell_1$ dual program $\xi_v(G)$ in \autoref{lem:l1-dual-vertex} of $\eta_v(G)$.

\begin{proposition}[Constructing $k$-Step Approximation] \label{prop:k-step-construction-vertex}
Let $G = (V, E, \pi)$ be a vertex-weighted directed graph.
Given a solution $f : V \to \R$ to $\xi_v(G)$ with objective value $\xi_f$,
there exists a $k$-step function $y: V \rightarrow \R$ with
\[
\norm{f - y}_{\pi}^2 \lesssim \frac{k \cdot \xi_f}{\vec{\sigma}_k^{v*}(G)}.
\] 
\end{proposition}

Combining the two propositions, using $\lambda_f = \xi_f$, and applying the dimension reduction result in \autoref{thm:dimension-reduction}, we get
\[
\vec{\psi}(G) 
\lesssim k \cdot \lambda_f + k \norm{f-y_f}_{\pi} \sqrt{\lambda_f}
\lesssim \frac{k^{1.5} \cdot \lambda_f}{\sqrt{\vec{\sigma}_k^{v*}(G)}}
\lesssim \frac{k^{1.5} \cdot \log(\Delta \cdot \alpha(G)) \cdot \vec{\lambda}_2^{v*}(G)  }{\sqrt{\vec{\lambda}_k^{v*}(G)}}.
\]

The proof of \autoref{prop:k-step-rounding-vertex} is by combining the arguments in \cite[Proposition 6.3]{KLT22} and \autoref{prop:l22-to-l1}.
The proof of \autoref{prop:k-step-construction-vertex} is essentially the same as in \cite[Proposition 6.3]{KLT22}.
There are no new steps in these proofs, 
so we omit them so as not to overload this paper.

\section{Cheeger-Type Inequalities for Hypergraphs} \label{sec:hypergraphs}

Louis~\cite{Lou15} and Chan, Louis, Tang, Zhang~\cite{CLTZ18} developed a spectral theory for hypergraphs based on a continuous time diffusion process.
They used it to derive a Cheeger inequality for hypergraph edge conductance, a higher-order Cheeger inequality for hypergraph $k$-way edge conductance, and a Cheeger inequality for hypergraph small-set conductance.

In this section, we will use the reweighted eigenvalue approach to derive similar results and compare with the results in~\cite{CLTZ18}.
In addition, we will prove an improved Cheeger inequality for hypergraph edge conductance, that was not known before.
Since the proofs of these results are all essentially the same as the corresponding proofs in~\cite{KLT22}, we just provide quick sketches so as to not overload this paper.

We note that vertex expansion in a hypergraph $H$ can simply be reduced to vertex expansion in its clique-graph $G$, and so the results of~\cite{KLT22} can be directly applied with $\Delta(G) \leq \Delta(H) \cdot r$, where $r$ is the maximum size of a hyperedge in $H$.
So we will only focus on hypergraph edge conductance in this section.

\subsection{Cheeger Inequality for Hypergraphs}

{\bf Results in~\cite{Lou15,CLTZ18}}: Given a hypergraph $H=(V,E)$ with a weight function $w: E \to \R_{\geq 0}$, they defined a nonlinear Laplacian operator and its eigenvalues $\gamma_1 \leq \gamma_2 \leq \ldots \leq \gamma_{|V|}$ based on a continuous time diffusion process.
Then they derived a Cheeger inequality $\frac12 \gamma_2 \leq \phi(H) \leq \sqrt{2 \gamma_2}$ in~\cite[Theorem 6.1]{CLTZ18},
where $\phi(H)$ is the hypergraph edge conductance of $H$ in \autoref{def:hypergraph-edge-conductance}.
But the quantity $\gamma_2$ in~\cite[Definition 3.1]{CLTZ18} is not polynomial time computable, 
and so a semidefinite programming relaxation $\tilde{\gamma}_2$ of $\gamma_2$ (see~\cite[SDP 8.3]{CLTZ18}) was used in~\cite[Theorem 8.1]{CLTZ18} to prove that 
\begin{equation} \label{e:Louis-Cheeger}
\tilde{\gamma}_2 \lesssim \phi(H) \lesssim \sqrt{ \tilde{\gamma}_2 \cdot \log r}.
\end{equation}

{\bf Our Results}: In the reweighted eigenvalue approach, we use $\gamma_2^*(H)$ in \autoref{def:hypergraph-edge-primal} as a relaxation to $\phi(H)$.
We can prove the easy direction as in \autoref{prop:directed-edge-easy} by a reduction, but we actually do not need to prove it as we will see soon.
As in \autoref{prop:lambda2-edge}, we can write $\gamma_2^*(H)$ as the following semidefinite program:
\begin{align*}
    \gamma_2^{*}(H) :=
     \min_{f: V \rightarrow \R^n} \max_{A \geq 0} &~~~\sum_{e\in E}\sum_{\{u,v\}\subseteq e} c(u,v,e)||f(u)-f(v)||^2
     \\
    \st&~~~
    \sum_{\{u,v\} \subseteq e} c(u,v,e) \leq w(e) \quad \quad \forall u,v \in V
    \\
    &~~~
    \sum_{v \in V} d_w(v)\cdot f(v) = \vec{0}
    \\
    &~~~ \sum_{v \in V} d_w(v) \cdot \norm{f(v)}^2 = 1.
\end{align*}
By using LP duality in the inner maximization problem as in \autoref{lem:l1-dual-vertex}, it follows that
\begin{align*}
    \gamma_2^{*}(H) =
     \min_{f: V \rightarrow \R^n} \min_{g:V \to \R_{\geq 0}} &~~~\sum_{e\in E} g(e) \cdot w(e) 
     \\
    \st&~~~
    g(e) \geq \norm{f(u)-f(v)}^2 \quad \quad \forall \{u,v\} \subseteq e, \forall e \in E
    \\
    &~~~
    \sum_{v \in V} d_w(v)\cdot f(v) = \vec{0}
    \\
    &~~~ \sum_{v \in V} d_w(v) \cdot \norm{f(v)}^2 = 1.
\end{align*}
It turns out that $\gamma_2^*(H)$ in this form is exactly the same as $\tilde{\gamma}_2$ in~\cite[SDP 8.3]{CLTZ18}.
(Also, $\gamma_2$ in~\cite{CLTZ18} is simply this dual program restricted to one dimensional embeddings $f:V\to\R$ as stated in \autoref{lem:l1-dual-edge}.)
Therefore, \autoref{thm:hypergraph-edge-conductance} follows from their result in \eqref{e:Louis-Cheeger}.
We would like to mention that we initially proved \autoref{thm:hypergraph-edge-conductance} using the same proofs as in~\cite{KLT22}, which is not surprising as the proofs in~\cite{KLT22} are very similar to that in~\cite{LRV13,CLTZ18}.

\subsection{Higher-Order Cheeger Inequality and Small-Set Expansion for Hypergraphs}

{\bf Results in~\cite{Lou15,CLTZ18}}: Given a hypergraph $H=(V,E)$ with a weight function $w: E \to \R_{\geq 0}$, the $k$-way edge conductance of $H$ is defined as
$\phi_k(H) := \min_{S_1, S_2, \ldots, S_k} \max_{1 \leq i \leq k} \phi(S_i)$ where $S_1, S_2, \ldots, S_k$ are over pairwise disjoint subsets of $V$.
\cite[Theorem 6.14]{CLTZ18} states that 
\begin{equation} \label{e:higher-order-Louis}
\tilde{\gamma}_k \lesssim \phi_{k}(H) \lesssim k^4 \cdot \log k \cdot \log \log k \cdot \log r \cdot \sqrt{\tilde{\gamma}_k}
{\rm~and~}
\phi_{(1-\eps)k}(H) \lesssim \frac{k^{2.5}}{\eps^{1.5}} \cdot \log k \cdot \log \log k \cdot \log r \cdot \sqrt{\tilde{\gamma}_k}
\end{equation}
for any $\eps \geq 1/k$, 
where $\tilde{\gamma}_k$ is an SDP relaxation of $\gamma_k$ which can be computed in polynomial time.
Furthermore, they proved a stronger bound in~\cite[Corollary 3.23]{CLTZ18} about small-set conductance that there is a subset $S$ with $|S| = \Theta(n/k)$ and
\begin{equation} \label{e:SSE-Louis}
\phi(S) \lesssim k^{1.5} \cdot \log k \cdot \log \log k \cdot \log r \cdot \sqrt{\tilde{\gamma}_k}.
\end{equation}

{\bf Our Results}:
We define $\gamma_k^*(H)$ as in \autoref{def:hypergraph-edge-primal} but the objective is to maximize the $k$-th smallest eigenvalue of the normalized Laplacian matrix $\L = I - D^{-1/2}AD^{-1/2}$.
This is, however, not a convex optimization problem.
As in~\cite{KLT22}, we change the objective to maximize the sum of the $k$ smallest eigenvalues $\sum_{i=1}^k \lambda_i(\L)$, so that we can use \autoref{prop:sum-of-lambda-k} to write this as a semidefinite program that we call $\sigma_k^*(H)$.
Using the same manipulations as in \cite[Proposition 5.2]{KLT22}, we can write
\begin{align*}
    \sigma_k^{*}(H) :=
     \min_{f: V \rightarrow \R^n} \min_{g:V \to \R_{\geq 0}} &~~~\sum_{e\in E} g(e) \cdot w(e) 
     \\
    \st&~~~
    g(e) \geq \norm{f(u)-f(v)}^2 \quad \quad \forall \{u,v\} \subseteq e, \forall e \in E
    \\
    &~~~
    \sum_{v \in V} d_w(v)\cdot f(v) f(v)^T \preceq I_n
    \\
    &~~~ \sum_{v \in V} d_w(v) \cdot \norm{f(v)}^2 = k.
\end{align*}
Following the same but rather long proof in~\cite[Section~5.2-5.4]{KLT22}, we can construct functions $f_1,\ldots,f_l$ with disjoint supports such that each is a good solution to \autoref{def:hypergraph-edge-primal}, and prove the exact same statement as in \cite[Theorem 5.20]{KLT22}.
\begin{theorem}[Higher-Order Cheeger Inequality for Hypergraphs] \label{thm:higher-order-hypergraphs}
For any hypergraph $H=(V,E)$ with weight function $w: E \to \R_{\geq 0}$,
\[
\frac{1}{k} \sigma_k^*(H) 
\lesssim \phi_k(H) 
\lesssim k^4 \log k \sqrt{\log r \cdot \sigma_k^*(H)}
\quad {\rm and} \quad
\phi_{(1-\eps)k}(H) 
\lesssim \frac{1}{\eps^4} \log k \sqrt{\log r \cdot \sigma_{k}^*(H)}. 
\]
\end{theorem}
Compared to \eqref{e:higher-order-Louis},
the result for $k$-way partitioning is comparable with an extra factor of $\sqrt{k}$ but a factor of $\sqrt{\log r}$ less,
while the result for $[(1-\eps)k]$-way partitioning for constant $\eps$ is an improvement by a factor of more than $k^2$.
As a consequence, this also implies an improvement of \eqref{e:SSE-Louis} for small-set conductance by a factor of more than $k$.

\subsection{Improved Cheeger Inequality for Hypergraphs}

Using the reweighted eigenvalue approach, we can also prove an analog of the improved Cheeger's inequality as described in \autoref{sec:improved-Cheeger-directed}.
This is a new result that was not obtained in~\cite{Lou15,CLTZ18}.
Combining with the higher-order Cheeger inequality for hypergraphs in \autoref{thm:higher-order-hypergraphs}, this implies the following corollary that only depends on the combinatorial structure of $H$: If the $k$-way edge conductance $\phi_k(H)$ is large for a small $k$, then $\gamma_2^*(H)$ is a tighter approximation to $\phi(H)$.

\begin{theorem}[Improved Cheeger's Inequality for Hypergraphs]
\label{thm:improved-cheeger-hypergraphs}
Let $H = (V, E)$ be a hypergraph with weight function $w: E \rightarrow \R_{\ge 0}$.
For any $2 \le k \le n$,
\[
  \phi(H) \lesssim \frac{k^{3/2} \cdot \log r \cdot \gamma_2^*(H)}{\sqrt{\gamma_k^*(H)}}.
\]
\end{theorem}

The proof of \autoref{thm:improved-cheeger-hypergraphs} follows the same two-step approach as in~\cite{KLLOT13,KLT22} and also in \autoref{sec:improved-Cheeger-directed} in this paper.
Actually, the proof for hypergraphs is very similar to that for undirected vertex expansion in~\cite{KLT22}, and is easier than that for directed graphs in \autoref{sec:improved-Cheeger-directed}. 
So, we omit the details and refer the reader to \autoref{sec:improved-Cheeger-directed} for an overview.
The only new element is the formulation in \autoref{def:hypergraph-edge-primal}.

\section{Concluding Remarks}

In this paper, we show that the reweighted eigenvalue approach can be extended substantially to derive Cheeger inequalities for directed graphs and hypergraphs.
Most notably, this develops into an interesting new spectral theory for directed graphs, which is much closer to the spectral theory for undirected graphs than what are previously known.
We hope that this spectral theory will find more applications in practice, in clustering and partitioning of directed graphs and hypergraphs.

Technically, the reweighted eigenvalue approach provides an intuitive and unifying method to reduce the study of expansion properties in more general settings to the basic setting of edge conductance in undirected graphs.
We believe that this approach can be used to lift more results in spectral graph theory for undirected graphs to more general settings,
as the ideas are consistent with recent works on directed Laplacian solvers and hypergraph spectral sparsification that we mentioned in \autoref{sec:related-work}.

There are some concrete open problems.
The most obvious one is to prove tight bounds for the two main results \autoref{thm:directed-vertex-expansion} and \autoref{thm:directed-edge-conductance}, to settle whether the dependency on the asymmetric ratio can be completely removed or not\footnote{See \autoref{r:tight} that the dimension reduction result for directed edge conductance is tight, and so a positive result removing the $\log \alpha(G)$ factor in \autoref{thm:directed-edge-conductance} would probably need substantial new ideas.
We incline to believe that the $\log \alpha(G)$ factor in \autoref{thm:directed-edge-conductance} cannot be completely removed, but we do not have an example supporting this belief.
We are less sure about what the right bound should be for \autoref{thm:directed-vertex-expansion}.}.
Another one is to formulate and prove higher-order Cheeger inequality and bipartite Cheeger inequality for directed graphs as discussed in \autoref{sec:supporting-results}. 
An important one for applications is to design fast algorithms (ideally near-linear time algorithms) for computing reweighted eigenvalues.
\newpage

\bibliographystyle{plain}

\end{document}